\documentclass{vldb}

\newif\ifextended
\extendedtrue

\usepackage{algorithm}
\usepackage[noend]{algpseudocode}
\usepackage{amsmath}
\usepackage{amssymb}
\usepackage{bbm, dsfont}
\usepackage{stmaryrd}
\usepackage{url}
\usepackage{listings}
\usepackage{xcolor}
\usepackage{graphicx}
\usepackage{wrapfig}
\usepackage{caption}
\usepackage{subcaption}
\usepackage{mathrsfs}
\usepackage{esvect}
\usepackage{xspace}
\usepackage{array, multirow}
\usepackage{rotating, makecell}
\usepackage{enumitem}
\usepackage{tikz} \usetikzlibrary{patterns}
\usepackage{pgfplots}
\usepackage{filecontents}
\usepackage{comment}
\usepackage{balance}

\vldbTitle{Data Migration using Datalog Program Synthesis}
\vldbAuthors{Yuepeng Wang, Rushi Shah, Abby Criswell, Rong Pan, Isil Dillig}
\vldbDOI{https://doi.org/10.14778/xxxxxxx.xxxxxxx}
\vldbVolume{13}
\vldbNumber{xxx}
\vldbYear{2020}

\newtheorem{definition}{Definition}
\newtheorem{theorem}{Theorem}
\newtheorem{lemma}{Lemma}
\newtheorem{example}{Example}
\newcommand{\toolname}{\textsc{Dynamite}\xspace}
\newcommand{\toolvar}{\toolname-{\sc Enum}\xspace}
\newcommand{\mitra}{\textsc{Mitra}\xspace}
\newcommand{\migrator}{\textsc{Migrator}\xspace}
\newcommand{\trinity}{\textsc{Trinity}\xspace}
\newcommand{\neo}{\textsc{Neo}\xspace}
\newcommand{\eirene}{\textsc{Eirene}\xspace}

\newcommand{\mweaver}{\textsc{MWeaver}\xspace}

\newcommand{\zaatar}{Zaatar\xspace}
\newcommand{\alps}{\textsc{Alps}\xspace}

\newcommand{\set}[1]{\{ #1 \}}
\newcommand{\denot}[1]{\llbracket #1 \rrbracket}

\newcommand{\dom}{\emph{dom}}

\newcommand{\hole}{\textbf{??}}
\newcommand{\colondash}{~\texttt{:-}~}

\newcommand{\tint}{\tt Int}
\newcommand{\tstring}{\tt String}

\newcommand{\name}{N}
\newcommand{\attr}{a}

\newcommand{\schema}{\mathcal{S}}
\newcommand{\inst}{\mathcal{D}}
\newcommand{\ex}{\mathcal{E}}
\newcommand{\ein}{\mathcal{I}}
\newcommand{\eout}{\mathcal{O}}
\newcommand{\prog}{\mathcal{P}}

\newcommand{\attrMap}{\Psi}
\newcommand{\sketch}{\Omega}
\newcommand{\cstr}{\Phi}
\newcommand{\model}{\sigma}
\newcommand{\mdp}{\varphi}
\newcommand{\mdpSet}{\Delta}
\newcommand{\worklist}{\mathcal{W}}
\newcommand{\visited}{\mathcal{V}}
\newcommand{\attribs}{\emph{PrimAttrbs}}
\newcommand{\rel}{\mathcal{R}}

\newcommand{\irule}[2]%
   {\mkern-2mu\displaystyle\frac{#1}{\vphantom{,}#2}\mkern-2mu}
\newcommand{\irulelabel}[3]
{
\mkern-2mu
\begin{array}{ll}
\displaystyle\frac{#1}{\vphantom{,}#2} & \!\!\!\!#3
\end{array}
\mkern-2mu
}

\algdef{S}[FOR]{ForRepeat}[1]{\algorithmicrepeat\ #1\ \textbf{times}}

\colorlet{punct}{red!60!black}
\definecolor{delim}{RGB}{20,105,176}
\colorlet{numb}{magenta!60!black}

\lstdefinelanguage{json}{
    basicstyle=\normalfont\ttfamily,
    numbersep=8pt,
    showstringspaces=false,
    string=[s]{"}{"},
    literate=
     *{0}{{{\color{numb}0}}}{1}
      {1}{{{\color{numb}1}}}{1}
      {2}{{{\color{numb}2}}}{1}
      {3}{{{\color{numb}3}}}{1}
      {4}{{{\color{numb}4}}}{1}
      {5}{{{\color{numb}5}}}{1}
      {6}{{{\color{numb}6}}}{1}
      {7}{{{\color{numb}7}}}{1}
      {8}{{{\color{numb}8}}}{1}
      {9}{{{\color{numb}9}}}{1}
      {:}{{{\color{punct}{:}}}}{1}
      {,}{{{\color{punct}{,}}}}{1}
      {\{}{{{\color{delim}{\{}}}}{1}
      {\}}{{{\color{delim}{\}}}}}{1}
      {[}{{{\color{delim}{[}}}}{1}
      {]}{{{\color{delim}{]}}}}{1},
}

\begin{document}


\title{Data Migration using Datalog Program Synthesis}


%

\numberofauthors{3}

\author{
\alignauthor
Yuepeng Wang \\
\affaddr{University of Texas at Austin}
\email{ypwang@cs.utexas.edu}
\alignauthor
Rushi Shah \\
\affaddr{University of Texas at Austin}
\email{rshah@cs.utexas.edu}
\alignauthor
Abby Criswell \\
\affaddr{University of Texas at Austin}
\email{abby@cs.utexas.edu}
\and
\alignauthor
Rong Pan \\
\affaddr{University of Texas at Austin}
\email{rpan@cs.utexas.edu}
\alignauthor
Isil Dillig \\
\affaddr{University of Texas at Austin}
\email{isil@cs.utexas.edu}
}

\maketitle

\begin{abstract}
This paper presents a new technique for migrating data between different schemas. Our method expresses the schema mapping as a Datalog program and automatically synthesizes a Datalog program from simple input-output examples to perform data migration. This approach  can  transform data between different types of schemas (e.g., relational-to-graph, document-to-relational) and performs  synthesis efficiently by leveraging the semantics of Datalog. We implement the proposed technique as a  tool called \toolname and show its effectiveness by evaluating \toolname  on 28 realistic data migration scenarios.
\end{abstract}

\section{Introduction} \label{sec:intro}

A prevalent task in today's ``big data era'' is the need to transform data stored in a source schema to a different target schema. For example, 
this task arises frequently when parties need to exchange or integrate data that are stored in different formats. In addition, 
as the needs of businesses evolve over time, it may become necessary to change the schema of the underlying database or move to a different type of database altogether. For instance, there are several real-world scenarios that necessitate shifting from a relational database to a non-SQL database or vice versa.

In this paper, we present a new programming-by-example technique for automatically migrating data from one schema to another. Given a small input-output example illustrating the source and target data, our method automatically synthesizes a program that  transforms data in the source format to its corresponding target format. Furthermore, unlike prior programming-by-example efforts in this space~\cite{glav-sigmod11, mweaver-sigmod12, mitra-vldb18}, our method can  transform data between several  types of database schemas, such as from a graph database to a relational one or  from a SQL database to a JSON document.

One of the key ideas underlying our method is to reduce the automated data migration problem to that of synthesizing a Datalog program from examples. Inspired by the similarity between Datalog rules and popular schema mapping formalisms, such as GLAV~\cite{glav-sigmod11, Fagin-tcs05} and tuple-generating dependencies~\cite{clio-vldb02},  our method expresses the correspondence between the source and target schemas  as a Datalog program in which extensional relations define the source schema and intensional relations represent the target. Then, given an input-output example $(\ein, \eout)$, finding a suitable schema mapping boils down to inferring a Datalog program $\prog$ such that $(\ein, \eout)$ is a model of $\prog$. Furthermore, because a Datalog program is executable, we can automate the  data migration task by simply executing the synthesized Datalog program $\prog$ on the source instance.

\begin{figure}[!t]
\centering
\includegraphics[scale=0.52]{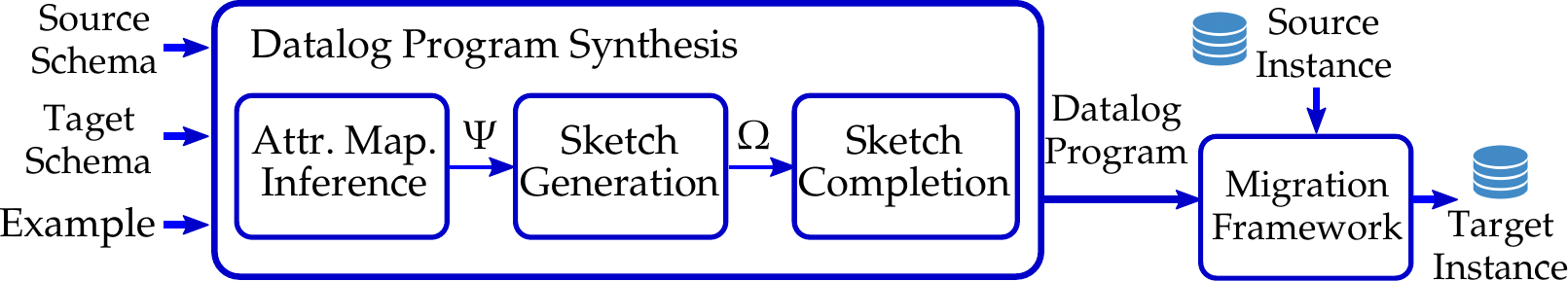}
\vspace{3pt}
\caption{Schematic workflow of \toolname.}
\vspace{-15pt}
\label{fig:workflow}
\end{figure}

While we have found Datalog programs to be a natural fit for expressing data migration tasks that arise in practice,  \emph{automating} Datalog program synthesis turns out to be a challenging task for several reasons: First, without some a-priori knowledge about the underlying schema mapping, it is unclear what the structure of the Datalog program would look like. Second, even if we ``fix'' the general structure of the Datalog rules, the  search space over all possible Datalog programs is still very large. Our method deals with these challenges by employing a practical  algorithm that leverages both the  semantics of Datalog programs as well as our target application domain.
As shown schematically in Figure~\ref{fig:workflow}, our proposed synthesis algorithm consists of three steps:

\vspace{5pt}
\noindent \emph{\textbf{Attribute mapping inference.}}
The first step of our approach is to infer an \emph{attribute mapping} $\attrMap$ which  maps each attribute in the source schema to a \emph{set} of  attributes that it {may}  correspond to.  While this attribute mapping does not uniquely define how to transform the source database to the target one, it substantially constrains the space of possible Datalog programs that we need to consider.

\vspace{5pt}
\noindent \emph{\textbf{Sketch generation.}}
In the next step, our method leverages the inferred attribute mapping $\attrMap$ to express the search space of all possible schema mappings as a \emph{Datalog program sketch} where some of the arguments of the extensional relations are unknown. While such a sketch represents a \emph{finite}  search space,  this space is exponentially large, making it infeasible to naively enumerate all programs defined by the sketch.

\vspace{5pt}
\noindent \emph{\textbf{Sketch completion.}}
The  final and most crucial ingredient of our method is the \emph{sketch completion} step that 
performs Datalog-specific deductive reasoning to 
dramatically prune the search space. Specifically, given a Datalog program  that does not satisfy the input-output examples, our method performs logical inference to rule out \emph{many other} Datalog programs from the search space. In particular, our method leverages a semantics-preserving transformation  as well as a new concept called \emph{minimal distinguishing projection} (MDP) to generalize from one incorrect Datalog program to many others.


\vspace{5pt}
\noindent \emph{\textbf{Results.}}
We have implemented our proposed technique in a prototype called \toolname and evaluate it on $28$ data migration tasks between real-world data-sets. These tasks involve transformations between different types of source and target schemas, including  relational, document, and graph databases. Our experimental results show that \toolname can successfully automate all of these tasks using small input-output examples that consist of just a few records. Furthermore, our method performs synthesis quite fast (with an average of $7.3$ seconds per benchmark) and can be used to migrate real-world database instances to the target schema in an average of $12.7$ minutes per database.

\vspace{5pt}
\noindent \emph{\textbf{Contributions.}}
The contributions of this paper include:
\vspace{-5pt}
\begin{itemize}[leftmargin=*]\itemsep-.5em
\item a  formulation of the automated data migration problem in terms of Datalog program synthesis;
\item a new algorithm for synthesizing  Datalog programs;
\item an implementation of our technique in a tool called \toolname and experimental evaluation  on $28$ data exchange tasks between real-world data-sets.
\end{itemize}

\section{Overview} \label{sec:overview}

In this section, we give a high-level overview of our method using a simple motivating example. Specifically, consider a document database with  the following schema:
\vspace{-5pt}
\begin{lstlisting}[xleftmargin=10pt,language=json,basicstyle=\normalfont\ttfamily\small]
Univ: [{ id: Int, name: String,
         Admit: [{uid: Int, count: Int}] }]
\end{lstlisting}
\vspace{-5pt}
This  database stores  a list of  universities, where each university has its own id, name, and graduate school admission information. Specifically, the admission information consists of a university identifier and the number of undergraduate students admitted from that university. 

Now, suppose that we need to transform this data to the following alternative schema:

\vspace{-5pt}
\begin{lstlisting}[language=json,basicstyle=\normalfont\ttfamily\small]
Admission:[{grad: String, ug: String, num: Int}]
\end{lstlisting}
\vspace{-3pt}

This new schema stores admission information as tuples consisting of a graduate school {\tt grad}, an undergraduate school {\tt ug}, and an integer {\tt num} that indicates the number of undergraduates from {\tt ug} that went to graduate school at {\tt grad}. As an example, Figure \ref{fig:example-univ}(a) shows a small subset of the data in the source schema, and \ref{fig:example-univ}(b) shows its corresponding representation in the target schema.

For this example, the desired transformation from the source to the target schema can be represented using the following simple Datalog program:
\vspace{-4pt}
\[
\begin{array}{l}
\emph{Admission}(\emph{grad}, \emph{ug}, \emph{num}) \colondash \\
\qquad \emph{Univ}(\emph{id}_1, \emph{grad}, v_1), \emph{Admit}(v_1, \emph{id}_2, \emph{num}), \emph{Univ}(\emph{id}_2, \emph{ug}, \_). \\
\end{array}
\vspace{-2pt}
\]
Here, the relation \emph{Univ} corresponds to a university entity in the source schema, and the relation \emph{Admit} denotes its nested \emph{Admit} attribute.  In the body of the Datalog rule, the third argument of the first \emph{Univ} occurrence has the same first argument as \emph{Admit}; this indicates that $(\emph{id}_2, \emph{num})$ is nested inside the university entry $(\emph{id}_1, \emph{grad})$. Essentially, this Datalog rule says the following: ``If there exists a pair of universities with identifiers $\emph{id}_1, \emph{id}_2$ and names $\emph{grad}, \emph{ug}$ in the source document, and if $(\emph{id}_2, num)$ is a nested attribute of $\emph{id}_1$, then there should be an \emph{Admission} entry $(\emph{grad}, \emph{ug}, \emph{num})$ in the target database.''

In what follows, we  explain how \toolname synthesizes the above Datalog program  given just the source and target schemas and the input-output example from Figure \ref{fig:example-univ}.

\begin{figure}[!t]
\centering
\begin{subfigure}[t]{0.23\textwidth}
\centering
\begin{lstlisting}[xleftmargin=20pt,language=json,basicstyle=\normalfont\ttfamily\scriptsize]
Univ: [
 {id:1, name:"U1",
  Admit: [
   {uid:1, count:10},
   {uid:2, count:50}]},
 {id:2, name:"U2",
  Admit: [
   {uid:2, count:20},
   {uid:1, count:40}]}
]
\end{lstlisting}
\vspace{-10pt}
\caption{Input Documents}
\end{subfigure}
\begin{subfigure}[t]{0.23\textwidth}
\centering
\begin{lstlisting}[xleftmargin=20pt,language=json,basicstyle=\normalfont\ttfamily\scriptsize]
Admission: [
 {grad:"U1",
  ug:"U1", num:10},
 {grad:"U1",
  ug:"U2", num:50},
 {grad:"U2",
  ug:"U2", num:20},
 {grad:"U2",
  ug:"U1", num:40}
]
\end{lstlisting}
\vspace{-10pt}
\caption{Output Documents}
\end{subfigure}
\vspace{5pt}
\caption{Example database instances.}
\vspace{-15pt}
\label{fig:example-univ}
\end{figure}

\vspace{5pt}
\noindent \textbf{Attribute Mapping.} As mentioned in Section~\ref{sec:intro}, our method starts by inferring an \emph{attribute mapping} $\attrMap$, which specifies which attribute in the source schema \emph{may} correspond to which other attributes (either in the source or target). For instance, based on the  example provided in Figure~\ref{fig:example-univ},  \toolname infers the following attribute mapping $\attrMap$:
\vspace{-6pt}
\[
\begin{array}{rclrcl}
\emph{id} & \rightarrow & \set{\emph{uid}} &
\qquad \emph{name} & \rightarrow & \set{\emph{grad}, \emph{ug}} \\
\emph{uid} & \rightarrow & \set{\emph{id}} &
\qquad \emph{count} & \rightarrow & \set{\emph{num}} \\
\end{array}
\vspace{-6pt}
\]
Since the values stored in the \emph{name} attribute of \emph{Univ} in the source schema are the same as the values stored in the \emph{grad} and \emph{ug} attributes of the target schema, $\attrMap$ maps source attribute \emph{name} to \emph{both} target attributes \emph{grad} and \emph{ug}. Observe that our inferred attribute mapping can also map source attributes to other source attributes. For example, since the values in the \emph{id} field of \emph{Univ} are the same as the values stored in the nested \emph{uid} attribute, $\attrMap$ also maps $\emph{id}$ to $\emph{uid}$ and vice versa.

\vspace{5pt}
\noindent \textbf{Sketch Generation.}
In the next step, \toolname uses the inferred attribute mapping $\attrMap$ to generate a program sketch $\sketch$ that defines the search space over all possible Datalog programs that we need to consider. Towards this goal, we introduce an extensional (resp. intensional) relation for each document in the source (resp. target) schema, including relations for nested documents.  In this case, there is a single intensional relation $\emph{Admission}$ for the target schema; thus, we introduce the following single Datalog rule sketch with the \emph{Admission} relation as its head:
\begin{equation}
\begin{array}{l}
\emph{Admission}(grad, ug, num) \colondash \\
\qquad \emph{Univ}(\hole_1, \hole_2, v_1), \emph{Admit}(v_1, \hole_3, \hole_4), \\
\qquad \emph{Univ}(\hole_5, \hole_6, \_), \emph{Univ}(\hole_7, \hole_8, \_).
\end{array}
\label{eq:sketch}
\end{equation}
\vspace{-5pt}
\[
\begin{array}{c}
\hole_1, \hole_3, \hole_5, \hole_7  \in  \set{\emph{id}_1, \emph{id}_2, \emph{id}_3, \emph{uid}_1}, \ \ \hole_4  \in  \set{\emph{num}, \emph{count}_1} \\
\hole_2, \hole_6, \hole_8  \in  \set{grad, ug, \emph{name}_1, \emph{name}_2, \emph{name}_3} \\
\end{array}
\]
Here, $\hole_i$ represents a \emph{hole} (i.e., unknown) in the sketch, and its domain is indicated as $\hole_i \in \set{e_1, \ldots, e_n}$,  meaning that hole $\hole_i$ can be instantiated with an element drawn from $\set{e_1, \ldots, e_n}$. To see where this sketch is coming from, we make the following observations:
\vspace{-5pt}
\begin{itemize}[leftmargin=*]\itemsep0em
\item According to $\attrMap$, the \emph{grad} attribute in the target schema comes from the \emph{name} attribute of \emph{Univ}; thus, we must have an occurrence of \emph{Univ} in the rule body.
\item Similarly, the  \emph{ug} attribute in the target schema comes from the \emph{name} attribute of \emph{Univ} in the source; thus, we may need another  occurrence of \emph{Univ} in the body.
\item Since the \emph{num} attribute  comes from the \emph{count} attribute in the nested \emph{Admit} document, the body of the Datalog rule contains $\emph{Univ}(\hole_1, \hole_2, v_1), \emph{Admit}(v_1, \hole_3, \hole_4)$
denoting  an \emph{Admit} document stored inside \emph{some} \emph{Univ} entity (the nesting relation is indicated through variable $v_1$).
\item The domain of each hole is determined by $\attrMap$ and the number of occurrences of each relation in the Datalog sketch. For example, since there are three occurrences of \emph{Univ}, we have three variables $\emph{id}_1, \emph{id}_2, \emph{id}_3$ associated with the \emph{id} attribute of  \emph{Univ}. The domain of hole $\hole_1$ is given by  $\{ \emph{id}_1, \emph{id}_2, \emph{id}_3, \emph{uid}_1 \}$ because it refers to the \emph{id} attribute of \emph{Univ}, and \emph{id} may be an ``alias'' of $\emph{uid}$ according to $\Psi$.
\end{itemize}

\vspace{5pt}
\noindent \textbf{Sketch Completion.} While the Datalog program sketch $\sketch$ given above looks quite simple, it actually has $64,000$ possible completions; thus,  a brute-force enumeration strategy is intractable. To solve this problem, \toolname utilizes a novel sketch completion algorithm that aims to learn from failed synthesis attempts.  Towards this goal, we encode all possible completions of sketch $\sketch$ as a satisfiability-modulo-theory (SMT) constraint $\cstr$ where each model of $\cstr$ corresponds to a possible completion of $\sketch$.  For the sketch from Equation~\ref{eq:sketch}, our SMT encoding is the following formula $\cstr$:
\vspace{-5pt}
\[
\begin{array}{rl}
& (x_1 = \emph{id}_1 \lor x_1 = \emph{id}_2 \lor x_1 = \emph{id}_3 \lor x_1 = \emph{uid}_1) \\
\land & (x_2 = \emph{grad} \lor x_2 = \emph{ug} \lor x_2 = \emph{name}_1 \lor \ldots \lor x_2 = \emph{name}_3) \\
\land & \ldots ~ \land (x_8 = \emph{grad} \lor x_8 = \emph{ug} \lor \ldots \lor x_8 = \emph{name}_3) \\
\end{array}
\vspace{-5pt}
\]
Here, for each hole $\hole_i$ in the sketch, we introduce a variable $x_i$ and stipulate that $x_i$ must be instantiated with exactly one of the elements in its domain.\footnote{In the SMT encoding, one should think of $\emph{id}_1, \emph{id}_2$ etc. as constants rather than variables.}  Furthermore, since Datalog  requires  all variables in the  head to occur in the rule body, we  also conjoin the following constraint with $\cstr$ to enforce this requirement:
\vspace{-5pt}
\[
(x_2 = \emph{grad} \lor x_6 = \emph{grad}) \land (x_2 = \emph{ug} \lor x_6 = \emph{ug}) \land (x_4 = \emph{num})
\vspace{-5pt}
\]
Next, we query the SMT solver for a model of this formula. In this case, one possible model $\model$ of $\cstr$ is:
\vspace{-5pt}
\begin{equation} \label{eq:model}
\begin{array}{rl}
& x_1 = \emph{id}_1 \land x_2 = \emph{grad} \land x_3 = \emph{id}_1 \land x_4 = \emph{num} \\
\land & x_5 = \emph{id}_1 \land x_6 = \emph{ug} \land x_7 = \emph{id}_2 \land x_8 = \emph{name}_1 \\
\end{array}
\vspace{-5pt}
\end{equation}
which corresponds to the following Datalog program $\prog$:
\vspace{-5pt}
\[
\begin{array}{c}
\emph{Admission}(\emph{grad}, \emph{ug}, \emph{num}) \colondash \ \ 
 \emph{Univ}(\emph{id}_1, \emph{grad}, v_1), \\ \emph{Admit}(v_1, \emph{id}_1, \emph{num}), \ \ 
 \emph{Univ}(\emph{id}_1, \emph{ug}, \_), \emph{Univ}(\emph{id}_2, \emph{name}_1, \_). 
\end{array}
\vspace{-5pt}
\]
\begin{figure}[!t]
\centering
\vspace{10pt}
\begin{subfigure}[t]{0.2\textwidth}
\centering
\scriptsize
\begin{tabular}{|c|c|c|}
\hline
grad & ug & num \\
\hline
U1 & U1 & 10 \\
\hline
U2 & U2 & 20 \\
\hline
\end{tabular}
\caption{Actual Result}
\end{subfigure}
~
\begin{subfigure}[t]{0.2\textwidth}
\centering
\scriptsize
\begin{tabular}{|c|c|c|}
\hline
grad & ug & num \\
\hline
U1 & U1 & 10 \\
\hline
U1 & U2 & 50 \\
\hline
U2 & U2 & 20 \\
\hline
U2 & U1 & 40 \\
\hline
\end{tabular}
\caption{Expected Result}
\end{subfigure}
\caption{Actual and expected results of program $\prog$.}
\label{fig:example-output}
\vspace{-15pt}
\end{figure}

However, this program does not satisfy the user-provided example because evaluating it on the input  yields a result that is different from the expected one (see Figure~\ref{fig:example-output}).

Now, in the next iteration, we want the SMT solver to return a model that corresponds to a different Datalog program. Towards this goal, one possibility would be to conjoin the negation of  $\model$ with our SMT encoding, but this would rule out just a \emph{single} program in our search space. To make synthesis more tractable, we instead {analyze} the root cause of failure and try to infer other Datalog programs that also do not satisfy the examples.

To achieve this goal, our sketch completion algorithm leverages two key insights: First, given a Datalog program $\prog$, we can obtain a \emph{set} of semantically equivalent Datalog programs by renaming variables in an equality-preserving way. Second, since our goal is to rule out incorrect (rather than just semantically equivalent) programs,  we can further enlarge the set of rejected Datalog programs by performing root-cause analysis. Specifically, we express the root cause of incorrectness as a 
 \emph{minimal distinguishing projection (MDP)}, which is a minimal set of attributes that distinguishes the expected output from the actual output. For instance, consider the actual and expected outputs $\eout$ and $\eout'$ shown in Figure~\ref{fig:example-output}. An MDP for this example is the singleton \emph{num} because taking the projection of $\eout$ and $\eout'$ on \emph{num} yields different results.



Using these two key insights, our sketch completion algorithm infers $720$ other Datalog programs that are guaranteed \emph{not} to satisfy the input-output example and represents them using the following SMT formula:
\vspace{-5pt}
\begin{equation} \label{eq:blocking1}
\begin{array}{c}
(x_4 = \emph{num} \land x_1 \neq x_2 \land x_1 = x_3 \land x_1 \neq x_4 \land x_1 = x_5 \\
\land x_1 \neq x_6 \land x_1 \neq x_7 \land x_1 \neq x_8 \land \cdots \land x_7 \neq x_8) \\
\end{array}
\vspace{-5pt}
\end{equation}
We can use the negation of this formula as a ``blocking clause'' by conjoining it with the SMT encoding and rule out many infeasible solutions at the same time.

After  repeatedly sampling models of the sketch encoding and adding blocking clauses as discussed above,  \toolname finally obtains the following model:
\vspace{-5pt}
\[
\begin{array}{c}
x_1 = \emph{id}_1 \land x_2 = \emph{grad} \land x_3 = \emph{id}_2 \land x_4 = \emph{num} \\
\land x_5 = \emph{id}_2 \land x_6 = \emph{ug} \land x_7 = \emph{id}_3 \land x_8 = \emph{name}_1 \\
\end{array}
\vspace{-5pt}
\]
which corresponds to the following Datalog program (after some basic simplification):
\vspace{-5pt}
\[
\begin{array}{l}
\emph{Admission}(\emph{grad}, \emph{ug}, \emph{num}) \colondash \\
\qquad \emph{Univ}(\emph{id}_1, \emph{grad}, v_1), \emph{Admit}(v_1, \emph{id}_2, \emph{num}), \emph{Univ}(\emph{id}_2, \emph{ug}, \_). \\
\end{array}
\vspace{-5pt}
\]
This  program is  consistent with the provided examples and can  automate the desired data migration task.

\section{Preliminaries} \label{sec:prelim}

In this section, we  review some preliminary information on Datalog and our schema representation; then, we  explain how to represent data migration programs in  Datalog.

\subsection{Schema Representation}

We represent database schemas using (non-recursive) record types, which are general enough to express a wide variety of database schemas, including XML and JSON documents and graph databases.
Specifically, 
a  schema $\schema$ is a mapping from type names $N$ to their definition:
\[
\begin{array}{rcl}
\emph{Schema} \ \schema &::=& N \rightarrow T \\
\emph{Type} \ T &::=& \tau \ | \ \{ N_1, \ldots, N_n \} \\
\end{array}
\]
A type definition is either  a primitive type $\tau$ or a set of named attributes $\{ N_1, \ldots N_k \}$, and the type of  attribute $N_i$ is given by the schema $\schema$.  An attribute $N$ is a \emph{primitive attribute} if $\schema(N) = \tau$ for some primitive type $\tau$. Given a schema $\schema$, we write $\attribs(\schema)$ to denote all primitive attributes in $\schema$, and we write $\emph{parent}(N) = N'$ if  $N \in \schema(N') $.

\vspace{-5pt}
\begin{example}
Consider  the JSON document schema from our motivating example in Section~\ref{sec:overview}:
\begin{lstlisting}[xleftmargin=10pt,language=json,basicstyle=\normalfont\ttfamily\scriptsize]
Univ: [{ id: Int, name: String,
         Admit: [{uid: Int, count: Int}] }]
\end{lstlisting}
In our representation, this schema is represented as follows:
\[
\small
\begin{array}{c}
\schema(\emph{Univ}) = \set{{\tt id}, {\tt name}, {\tt Admit}} \quad
\schema(\emph{Admit}) = \set{{\tt uid}, {\tt count}} \\
\schema(\emph{id}) = \schema(\emph{uid}) = \schema(\emph{count}) = \tint \quad
\schema(\emph{name}) = \tstring \\
\end{array}
\]
\end{example}

\begin{example}

Consider the following relational schema:
\[
{\tt User}({\tt id}: \tint,~ {\tt name}: \tstring,~ {\tt address}: \tstring)
\]
In our representation, this corresponds to the following schema:
\[
\small
\begin{array}{c}
\schema(\emph{User}) = \set{{\tt id}, {\tt name}, {\tt address}} \\
\schema(\emph{id}) = \tint \quad
\schema(\emph{name}) = \schema(\emph{address}) = \tstring \\
\end{array}
\]
\end{example}

\begin{example}
{
Consider the following graph schema:
}
\begin{figure}[!h]
\centering
\includegraphics[scale=0.5]{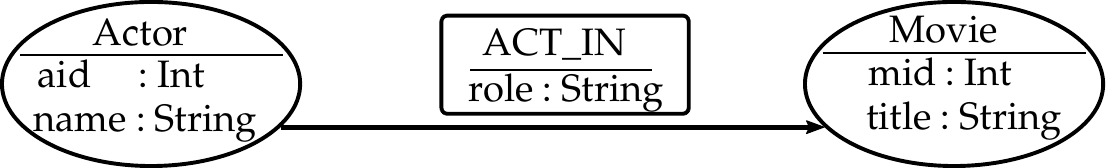}
\vspace{-10pt}
\end{figure}

To convert this schema to our representation, we first introduce two attributes {\tt source} and {\tt target} to denote the source and target nodes of the edge. Then, the graph schema corresponds the following mapping in our representation:
\[
\small
\begin{array}{c}
\schema({\tt Movie}) = \set{{\tt mid}, {\tt title}} \quad
\schema({\tt Actor}) = \set{{\tt aid}, {\tt name}} \\
\schema({\tt ACT\_IN}) = \set{{\tt source}, {\tt target}, {\tt role}} \\
\schema({\tt mid}) = \schema({\tt aid}) = \schema({\tt source}) = \schema({\tt target}) = \tint \\
\schema({\tt title}) = \schema({\tt name}) = \schema({\tt role}) = \tstring \\
\end{array}
\]
\end{example}

\subsection{Datalog} \label{sec:prelim-datalog}

\begin{figure}[!t]
\centering
\[
\begin{array}{rclrcl}
\emph{Program} &::=& \emph{Rule}^+ & \quad
\emph{Rule} &::=& \emph{Head} \colondash \emph{Body}. \\
\emph{Head} &::=& \emph{Pred} & \quad
\emph{Body} &::=& \emph{Pred}^+ \\
\emph{Pred} &::=& R(v^+) &
& & \hspace{-50pt} v \in \emph{Variable} \quad R \in \emph{Relation} \\
\end{array}
\]
\caption{Syntax of Datalog programs. }
\vspace{-10pt}
\label{fig:syntax-datalog}
\end{figure}

 As shown in Figure~\ref{fig:syntax-datalog}, a Datalog program consists of a list of rules, where each rule is of the form $H \colondash B.$ Here, $H$ is referred as the \emph{head} of the rule and $B$ is the \emph{body}. The head $H$ is a single relation of the form $R(v_1, \ldots, v_n)$, and the body $B$ is a collection of predicates $B_1, B_2, \ldots, B_n$. In the remainder of this paper, we sometimes also write
\[
    H_1, \ldots, H_m \colondash B_1, B_2, \ldots, B_n.
\]
as short-hand for $m$ Datalog rules with the same body.
Predicates that appear only in the body are known as extensional relations and correspond to known facts.  Predicates that appear in the head are called intensional relations and correspond to the output of the Datalog program.

\vspace{0.1in} \noindent
{\bf \emph{Semantics.}} 
The semantics of Datalog programs are typically given using  Herbrand models of first-order logic formulas~\cite{datalog-tkde89}. In particular, each Datalog rule $\mathcal{R}$ of the form $H(\vec{x}) \colondash B(\vec{x}, \vec{y})$ corresponds to a first-order formula $\denot{\mathcal{R}} =\forall \vec{x}, \vec{y}.~ B(\vec{x}, \vec{y}) \rightarrow H(\vec{x})$, and the semantics of the Datalog program can be expressed as the conjunction of each rule-level formula. Then, given a Datalog program $\prog$ and an input $\ein$  (i.e., a set of ground formulas), the output corresponds to the least Herbrand model of $\denot{\prog} \land \ein$.

\subsection{Data Migration using Datalog} \label{sec:prelim-migration}

We now  discuss how to perform  data migration using Datalog. The basic idea is as follows: First, given a source database instance $\inst$  over schema $\schema$, we  express $\inst$ as a collection of Datalog facts over extensional relations $\rel$. Then, we express the target schema $\schema'$ using intensional relations $\rel'$ and construct a set of (non-recursive) Datalog rules, one for each intensional relation in $\rel'$. Finally, we run this Datalog program and translate the resulting facts into the target database instance.
Since programs can be evaluated using an off-the-shelf Datalog solver, we only explain how to translate between database instances and Datalog facts.

\vspace{5pt}
\noindent \emph{\textbf{From instances to facts.}}
Given a database instance $\inst$ over schema $\schema$, we  introduce an extensional relation symbol $R_N$  for each record type with name $N$ in $\schema$ and assign a unique identifier $\emph{Id}(r)$ to every record $r$ in the database instance. Then, for each instance  $r = \set{a_1: v_1, \ldots, a_n: v_n}$ of  record type $N$, we generate a fact $R_N(c_0, c_1, \ldots, c_n)$ where:
\[
c_i =
\begin{cases}
\emph{Id}(\emph{parent}(r)),& \text{if } i = 0 \text{ and } $r$ \text{ is a nested record }\\
v_i,&\text{if $\schema(a_i)$ is a primitive type} \\
\emph{Id}(r),& \text{if $\schema(a_i)$ is a record type} \\
\end{cases}
\]

Intuitively, relation $R_N$ has an extra argument that keeps track of its parent record in the database instance if $N$ is nested in another record type. In this case, the first argument of $R_N$ denotes the unique identifier for the record in which it is nested. 
\vspace{-5pt}
\begin{example}
For the JSON document from  Figure~\ref{fig:example-univ}(a), our method  generates the following Datalog facts 
\vspace{-5pt}
\[
\small
\centering
\begin{array}{lll}
Univ(1, {\tt ``U1"}, id_1) & Univ(2, {\tt ``U2"}, id_2) &
Admit(id_1, 1, 10) \\ Admit(id_2, 2, 20) &
Admit(id_1, 2, 50) & Admit(id_2, 1, 40) \\
\end{array}
\vspace{-5pt}
\]
where $id_1$ and $id_2$ are unique identifiers.
\end{example}
\vspace{-5pt}
\noindent \emph{\textbf{From facts to instances.}}
We convert Datalog facts to the target database instance  using the inverse procedure. Specifically, given a Datalog fact $R_N(c_1, \ldots, c_n)$ for record type $N: \set{a_1, \ldots, a_n}$, we create a record instance using a function $\emph{BuildRecord}(R_N, N) = \set{a_1: v_1, \ldots, a_n: v_n}$ where
\[
v_i =
\begin{cases}
c_i,& \text{if $\schema(a_i)$ is a primitive type} \\
\emph{BuildRecord}(R_{a_i}, a_i),& \text{if $\schema(a_i)$ is a record type and} \\
& \text{the first argument of $R_{a_i}$ is $c_i$} \\
\end{cases}
\]

Observe that the \emph{BuildRecord} procedure builds the record  recursively by chasing parent identifiers into other relations.

\section{Datalog Program Synthesis} \label{sec:synthesis}

In this section, we describe our algorithm for automatically synthesizing Datalog programs from an input-output example $\ex = (\ein, \eout)$. Here, $\ein$ corresponds to an example of the database instance in the source schema, and $\eout$ demonstrates the desired target instance.  We  start by giving a high-level overview of the synthesis algorithm and then explain each of the key ingredients in more detail.

\subsection{Algorithm Overview} \label{sec:synth-top-level}

\begin{figure}[!t]
\begin{algorithm}[H]
\caption{Synthesizing Datalog programs}
\label{algo:synthesis}
\begin{algorithmic}[1]
\Procedure{\textsc{Synthesize}}{$\schema, \schema', \ex$}
\vspace{2pt}
\small
\Statex \textbf{Input:} Source schema $\schema$, target schema $\schema'$, \Statex \qquad \quad \ \ example $\ex = (\ein, \eout)$
\Statex \textbf{Output:} Datalog program $\prog$ or $\bot$ to indicate failure
\vspace{2pt}
\State $\attrMap \leftarrow \textsc{InferAttrMapping}(\schema, \schema', \ex)$;
\State $\sketch \leftarrow \textsc{SketchGen}(\attrMap, \schema, \schema')$;
\State $\cstr \leftarrow \textsc{Encode}(\sketch)$;
\While{\textsf{SAT}($\cstr$)}
    \State $\model \leftarrow \textsf{GetModel}(\cstr)$;
    \State $\prog \leftarrow \textsf{Instantiate}(\sketch, \model)$;
    \State $\eout' \leftarrow \denot{\prog}_{\ein}$;
    \If{$\eout' = \eout$} \Return $\prog$; \EndIf
    \State $\cstr \leftarrow \cstr \land \textsc{Analyze}(\model, \eout', \eout)$;
\EndWhile
\State \Return $\bot$;
\EndProcedure

\end{algorithmic}
\end{algorithm}
\vspace{-20pt}
\end{figure}

The top-level algorithm for synthesizing Datalog programs is summarized in Algorithm~\ref{algo:synthesis}. The \textsc{Synthesize} procedure takes as input a source schema $\schema$, a target schema $\schema'$, and an input-output example $\ex = (\ein, \eout)$. The return value is either  a Datalog program $\prog$ such that evaluating $\prog$ on $\ein$ yields $\eout$ (i.e. $\denot{\prog}_{\ein} = \eout$) or  $\bot$ to indicate that the desired data migration task cannot be represented as a Datalog program.

As shown in Algorithm~\ref{algo:synthesis}, the \textsc{Synthesize} procedure first invokes the \textsc{InferAttrMapping} procedure (line 2) to infer an attribute mapping $\attrMap$. Specifically, $\attrMap$ is a mapping from each  $a \in 
\attribs(\schema)$ to a set of attributes $\{ a_1, \ldots, a_n \}$ where $a_i \in \attribs(\schema) \cup  \attribs(\schema')$ such that:
\vspace{-5pt}
\[
a' \in \attrMap(a) \ \Leftrightarrow  \ \Pi_{a'}(\mathcal{D}) \subseteq \Pi_{a}(\ein)
\vspace{-5pt}
\]
where $\inst$ stands for either $\ein$ or $\eout$. Thus, \textsc{InferAttrMapping} is conservative and  maps a source attribute $a$ to another attribute $a'$ if the values contained in $a'$ are a subset of those contained in $a$.

Next, the algorithm invokes \textsc{SketchGen}  (line 3) to generate a Datalog program sketch $\sketch$ based on  $\attrMap$. As mentioned in Section~\ref{sec:overview}, a  sketch $\sketch$ is a Datalog program with unknown arguments in  the rule body, and the sketch also determines the domain for each unknown. Thus, if the sketch contains  $n$ unknowns, each with $k$ elements in its domain, then the sketch encodes a search space of $k^n$ possible  programs.

Lines 4-10 of the  \textsc{Synthesize}  algorithm perform lazy enumeration over possible sketch completions.   Given a sketch $\sketch$, we first generate an SMT formula $\cstr$ whose models correspond to  all possible completions of $\sketch$  (line 4).  Then, the loop in lines 5-10 repeatedly queries a model of $\cstr$ (line 6), tests if the corresponding Datalog program is consistent with the example (lines 7-9), and adds a blocking clause to $\cstr$ if it is not (line 10). The blocking clause is obtained via the call to the {\sc Analyze} procedure, which performs Datalog-specific deductive reasoning to infer a whole set of programs that are \emph{guaranteed not to} satisfy the examples.

In the remainder of this section, we explain the sketch generation and completion procedures in more detail.
\subsection{Sketch Generation} \label{sec:synth-sketch-gen}

Given an attribute mapping $\attrMap$, the goal of sketch generation is to construct the skeleton of the target Datalog program.
Our sketch language is similar to the Datalog syntax in Figure~\ref{fig:syntax-datalog}, except that it allows holes (denoted by $\hole$) as special constructs indicating unknown expressions.
As summarized in Algorithm~\ref{algo:sketch-gen}, the {\sc SketchGen} procedure iterates over each top-level record in the target schema and, for each record type, it generates a Datalog rule sketch  using the helper procedure {\sc GenRuleSketch}.
\footnote{
The property of the generated sketch is characterized and proved in\ifextended{
Appendix A.
}\else{
Appendix A of the extend version~\cite{extended-version}.
}\fi
}
Conceptually, {\sc GenRuleSketch} performs the following tasks: First, it generates a set of intensional predicates for each top-level record in the target schema (line 8). The intensional predicates do not contain any unknowns and only appear in the head of the Datalog rules.
     Next, the loop (lines 9--12) constructs the skeleton of each Datalog rule body  by generating extensional predicates for the relevant source record types. The extensional predicates do contain unknowns, and there can be multiple occurrences of a  relation symbol in the body.
     Finally, the loop in lines 13--18 generates the domain for each unknown used in the rule body.


\begin{figure}[!t]
\begin{algorithm}[H]
\caption{Generating Datalog program sketches}
\label{algo:sketch-gen}
\begin{algorithmic}[1]
\Procedure{\textsc{SketchGen}}{$\attrMap, \schema, \schema'$}
\vspace{2pt}
\small
\Statex \textbf{Input:} Attribute mapping $\attrMap$, source schema $\schema$,
\Statex \qquad \quad \ \ target schema $\schema'$
\Statex \textbf{Output:} Program sketch $\sketch$
\State $\sketch \leftarrow \emptyset$;
\For{\textbf{each} top-level record type $\name \in \schema'$}
    \State $R \leftarrow \textsc{GenRuleSketch}(\attrMap, \schema, \schema', N)$;
    \State $\sketch \leftarrow \sketch \cup \set{R}$;
\EndFor
\State \Return $\sketch$;
\EndProcedure
\vspace{5pt}

\Procedure{\textsc{GenRuleSketch}}{$\attrMap, \schema, \schema', \name$}
\vspace{2pt}
\small
\State $H \leftarrow \textsc{GenIntensionalPreds}(\schema', N)$; $B \leftarrow \emptyset$;
\vspace{2pt}
\For{\textbf{each} $\attr \in \dom(\attrMap)$}
    \ForRepeat{$| \set{\attr' ~|~ \attr' \in \emph{PrimAttrbs}(\name) \land \attr' \in \attrMap(\attr)} |$}
        \State $N \leftarrow \textsf{RecName}(\attr)$;
        \State $B \leftarrow B \cup \textsc{GenExtensionalPreds}(\schema, N)$;
    \EndFor
\EndFor
\vspace{2pt}
\For{\textbf{each} $\hole_a \in \emph{Holes}(B)$}
    \State $V \leftarrow \set{v_{a'} ~|~ a' \in \emph{PrimAttrbs}(\name) \land a' \in \attrMap(a)}$;
    \For{\textbf{each} $a' \in \attrMap(a) \cup \set{a}$ \textbf{and} $a' \in \emph{PrimAttrbs}(\schema)$}
        \State $n \leftarrow \textsf{CopyNum}(B, \textsf{RecName}(a'))$;
        \State $V \leftarrow V \cup \bigcup_{i=1}^{n} \set{v_{a'}^{i}}$;
    \EndFor
    \State $B \leftarrow B[\hole_a \mapsto \hole_a \in V]$;
\EndFor
\State \Return $H \colondash B.$;
\EndProcedure

\end{algorithmic}
\end{algorithm}
\vspace{-20pt}
\end{figure}

\noindent \textbf{\emph{Head generation.}} Given a top-level record type $N$ in the target schema, the procedure {\sc GenIntensionalPreds}  generates the head of the corresponding Datalog rule for $N$. If $N$ does not contain any nested records, then the head consists of a single predicate, but, in general, the head contains as many predicates as are (transitively) nested in $N$.

\begin{figure}[!t]
\centering
\[
\small
\hspace{-10pt}\begin{array}{c}
\irulelabel
{\schema'(N) \in \emph{PrimType}}
{\schema' \vdash N \leadsto (v_N,~ \emptyset)}
{\textrm{(InPrim)}} \\ \ \\

\irulelabel
{\begin{array}{c}
    \schema'(N) = \set{a_1, \ldots, a_n} \quad \emph{isNested}(N) \\
    \schema' \vdash a_i \leadsto (v_i, H_i) \quad i = 1, \ldots, n \\
\end{array}}
{\schema' \vdash N \leadsto (v_N,~ \set{R_N(v_N, v_1, \ldots, v_n)} \cup \bigcup_{i=1}^{n} H_i)}
{\textrm{(InRecNested)}} \\ \ \\

\irulelabel
{\begin{array}{c}
    \schema'(N) = \set{a_1, \ldots, a_n} \quad \neg \emph{isNested}(N) \\
    \schema' \vdash a_i \leadsto (v_i, H_i) \quad i = 1, \ldots, n \\
\end{array}}
{\schema' \vdash N \leadsto (\_,~ \set{R_N(v_1, \ldots, v_n)} \cup \bigcup_{i=1}^{n} H_i)}
{\textrm{(InRec)}}

\end{array}
\]
\caption{Inference rules describing {\sc GenIntensionalPreds}}
\label{fig:rules-intensional}
\vspace{-10pt}
\end{figure}

In more detail, Figure~\ref{fig:rules-intensional} presents the {\sc GenIntensionalPreds} procedure  as inference rules that derive judgments of the form $\schema' \vdash \name \leadsto (v, H)$ where $H$ corresponds to the head of the Datalog rule for record type $\name$. As expected, these rules are recursive and build the predicate set $H$ for $\name$ from those of its nested records. Specifically, given a top-level record $N$ with attributes $a_1, \ldots, a_n$, the rule \textrm{InRec} first generates predicates $H_i$ for each attribute $a_i$ and then introduces an additional relation $R_\name(v_1, \ldots, v_n)$  for $\name$ itself. Predicate generation for nested relations (rule \textrm{InRecNested}) is similar, but we introduce a new variable $v_N$ that is used for connecting $\name$ to its parent relation. The \textrm{InPrim} rule corresponds to the base case of  {\sc GenIntensionalPreds}  and generates variables for attributes of primitive type.

\vspace{5pt}
\noindent \textbf{\emph{Body sketch generation.}} We now consider sketch generation for the body of each Datalog rule (lines 9--12 in Algorithm~\ref{algo:sketch-gen}). Given a record type $N$ in the source schema and its corresponding predicate(s) $R_N$, the loop in lines 9--12 of Algorithm~\ref{algo:sketch-gen}  generates as many copies of $R_N$ in the rule body as there are head attributes that ``come from'' $R_N$ according to $\attrMap$. Specifically, Algorithm~\ref{algo:sketch-gen} invokes a procedure called {\sc GenExtensionalPreds}, described in Figure~\ref{fig:rules-extensional}, to generate each copy of the extensional predicate symbol.

Given a record type $N$ in the source schema,  \textsc{GenExtensionalPreds}  generates predicates up until the top-level record that contains $N$. The rules in Figure~\ref{fig:rules-extensional} are of the form $\schema \vdash N \hookrightarrow (h, B)$, where $B$ is the sketch body for record type $N$. The \textrm{ExPrim} rule is the base case to generate sketch holes for primitive attributes. Given a record $N$ with attributes $a_1, \ldots, a_n$ and its parent $N'$, the rule \textrm{ExRecNested} recursively generates the body predicates $B'$ for the parent record $N'$ and adds an additional predicate $R_N(v_N, h_1, \ldots, h_n)$ for $N$ itself. Here $v_N$ is a variable for connecting $N$ and its parent $N'$, and $h_i$ is the hole or variable for attribute $a_i$. In the case where $N$ is a top level record, the \textrm{ExRec} rule  generates a singleton predicate $R_N(h_1, \ldots, h_n)$.




\vspace{-5pt}
\begin{example} \label{ex:sketch-body}
Suppose we want to generate the body sketch for the rule associated with record type $T: \set{a': \tint,~ b': \tint}$ in the target schema. Also, suppose we are given the attribute mapping $\attrMap$ where $
\attrMap(a) = a'$ and  $\attrMap(b) = b'$ and source attributes $a, b$ belong to the following record type in the source schema:
$
C: \set{a: \tint,~ D: \set{b: \tint}}
$.
According to $\attrMap$, $a'$ comes from attribute $a$ of record type $C$ in the source schema, so we have a copy of $R_C$ in the sketch body. Based on the rules of Figure~\ref{fig:rules-extensional}, we generate predicate $R_C(\emph{\hole}_a, v_D^1)$, where $\emph{\hole}_a$ is the hole for attribute $a$ and $v_D^1$ is a fresh variable. Similarly, since  $b'$ comes from attribute $b$ of record type $D$, we generate predicates $R_C(\emph{\hole}_a, v_D^2)$ and $R_D(v_D^2, \emph{\hole}_b)$. Putting them together, we obtain the following sketch body:
\[
R_C(\emph{\hole}_a, v_D^1), R_C(\emph{\hole}_a, v_D^2), R_D(v_D^2, \emph{\hole}_b)
\]
\end{example}

\begin{figure}[!t]
\centering
\[
\small
\hspace{-10pt}\begin{array}{c}
\irulelabel
{\schema(N) \in \emph{PrimType}}
{\schema \vdash N \hookrightarrow (\hole_N,~ \_)}
{\textrm{(ExPrim)}} \\ \ \\

\irulelabel
{\begin{array}{c}
    \schema(N) = \set{a_1, \ldots, a_n} \quad \text{fresh } v_N \\
    \schema \vdash a_i \hookrightarrow (h_i, \_) \quad i = 1, \ldots, n \\
    N' = \emph{parent}(N) \quad \schema \vdash N' \hookrightarrow (\_, B') \\
\end{array}}
{\schema \vdash N \hookrightarrow (v_N,~ \set{R_N(v_N, h_1, \ldots, h_n)} \cup B')}
{\textrm{(ExRecNested)}} \\ \ \\

\irulelabel
{\begin{array}{c}
    \schema(N) = \set{a_1, \ldots, a_n} \quad \neg \emph{isNested}(N) \\
    \schema \vdash a_i \hookrightarrow (h_i, \_) \quad i = 1, \ldots, n \\
\end{array}}
{\schema \vdash N \hookrightarrow (\_,~ \set{R_N(h_1, \ldots, h_n)})}
{\textrm{(ExRec)}}

\end{array}
\]
\caption{Inference rules for {\sc GenExtensionalPreds}}
\vspace{-10pt}
\label{fig:rules-extensional}
\end{figure}

\vspace{-5pt}
\noindent \textbf{\emph{Domain generation.}}
Having constructed the  skeleton of the Datalog program, we still need to determine the set of variables that each hole in the sketch can be instantiated with. Towards this goal, the last part of  \textsc{GenRuleSketch}  (lines 13--18 in Algorithm~\ref{algo:sketch-gen}) constructs the domain $V$ for each hole as follows: First, for each attribute $a$ of source relation $R_N$, we introduce as many variables $v_a^1, \ldots, v_a^k$ as there are copies of $R_N$. Next, for the purposes of this discussion, let us say that attributes $a$ and $b$ ``alias'' each other if $b \in \attrMap(a)$ or vice versa. Then, given a hole $\hole_x$ associated with attribute $x$, the domain of $\hole_x$ consists of all the variables associated with attribute $x$ or one of its aliases.
\vspace{-5pt}
\begin{example}
Consider the same schemas and attribute mapping from Example~\ref{ex:sketch-body} and the following body sketch:
\[
R_C(\emph{\hole}_a, v_D^1), R_C(\emph{\hole}_a, v_D^2), R_D(v_D^2, \emph{\hole}_b)
\]
Here, we have $\emph{\hole}_a \in \set{v_{a'}, v_a^1, v_a^2}$ and  $\emph{\hole}_b \in \set{v_{b'}, v_b^1}$.
\end{example}


\subsection{Sketch Completion}\label{sec:sketch-completion}
While the sketch generated by Algorithm~\ref{algo:sketch-gen} defines a finite search space of Datalog programs, this search space is still exponentially large. Thus, rather than performing naive brute-force enumeration, our sketch completion algorithm combines enumerative search with Datalog-specific deductive reasoning to learn from failed synthesis attempts. 
 As explained in Section~\ref{sec:synth-top-level}, the basic idea is to generate an SMT encoding of all possible sketch completions and then iteratively add blocking clauses to  rule out incorrect Datalog programs. In the remainder of this section, we discuss how to generate the initial SMT encoding as well as the {\sc Analyze} procedure for generating useful blocking clauses. 

\vspace{5pt}
\noindent \emph{\textbf{Sketch encoding.}}
Given a Datalog program sketch $\sketch$, our initial SMT encoding is constructed as follows: First, for each  hole $\hole_i$ in the sketch, we introduce an integer variable $x_i$, and for every variable $v_j$ in the domain of some hole, we introduce a unique§ integer constant denoted as \emph{Const}($v_j$). Then, our SMT encoding  stipulates the following constraints to enforce that the sketch completion is well-formed:
\vspace{-5pt}
\begin{itemize}[leftmargin=*]\itemsep0em
\item
{\bf \emph{Every hole must be instantiated:}}  For each hole of the form $\hole_i \in \set{v_1, \ldots, v_n}$, we add a constraint
\[
\vspace{-0.1in}
    \bigvee\limits_{j=1}^{n} x_i = \emph{Const}(v_j)
\vspace{-3pt}
\]
\item {\bf \emph{Head variables must appear in the body.} } In a well-formed Datalog program, every head variable must appear in the body. Thus, for each head variable $v$,
 we add:
\[
    \bigvee\limits_{i} x_i = \emph{Const}(v) \text{ where } v \text{ is in the domain of } \hole_i
\]
\end{itemize}
\vspace{-8pt}
Since there is a one-to-one mapping between integer constants in the SMT encoding and sketch variables, each model of the SMT formula corresponds to a  Datalog program.

\vspace{5pt}
\noindent \emph{\textbf{Adding blocking clauses.}} Given a Datalog program $\prog$ that does not satisfy the examples $(\ein, \eout)$, our top-level synthesis procedure (Algorithm~\ref{algo:synthesis}) invokes a function called {\sc Analyze} to find useful blocking clauses to add to the SMT encoding. This  procedure is summarized in Algorithm~\ref{algo:analyze} and  is built on two key insights. The first key insight is that  the semantics of a Datalog program is unchanged under an equality-preserving renaming of variables:

\vspace{-5pt}
\begin{theorem} \label{thm:substitution}
Let $\prog$ be a Datalog program over variables $X$ and let $\hat{\sigma}$ be an injective substitution from $X$ to another set of variables $Y$. Then, we have $\prog \simeq \prog \hat{\sigma}$.
\end{theorem}
\vspace{-1.5em}
\begin{proof}
\ifextended{
See Appendix A.
}\else{
See Appendix A of the extended version~\cite{extended-version}.
}\fi
\end{proof}
\vspace{-5pt}

To see how this theorem is useful, let $\sigma$ be a model of our SMT encoding. In other words, $\model$ is a mapping from holes in the Datalog sketch to variables $V$. Now, let $\hat{\sigma}$ be an injective renaming of variables in $V$. Then, using the above theorem,  we know that any other assignment $\model' = \model \hat{\sigma}$ is also guaranteed to result in an incorrect Datalog program.

Based on this insight, we can generalize from the specific assignment $\model$ to a more general class of incorrect assignments as follows: If a hole is not assigned to a head variable, then it can be assigned to any variable in its domain as long as it respects the equalities and disequalities in  $\model$. Concretely, given  assignment $\model$, we  generalize it as follows:
\[
\emph{Generalize}(\model) = \bigwedge_{x_i \in \emph{dom}(\model)} \alpha(x_i, \model), \ \textrm{where}
\]
\vspace{-10pt}
\[
\alpha(x, \model) = \left \{ 
\begin{array}{ll}
x = \model(x)  &  \text{if} \ \model(x) \text{ is a head variable} \\
\bigwedge_{x_j \in \emph{dom}(\model)} x \star x_j
& \text{otherwise}
\end{array}
\right .
\]
Here the binary operator $\star$ is defined to be equality if $\model$ assigns both $x$ and $x_j$ to the same value, and disequality otherwise. Thus, rather than  ruling out just the current assignment $\model$, we can instead use $\neg \emph{Generalize}(\model)$ as a much more general blocking clause that rules out several equivalent Datalog programs at the same time.

\vspace{-5pt}
\begin{example} \label{ex:generalize}
Consider again the sketch from Section~\ref{sec:overview}:
\[
\begin{array}{c}
\emph{Admission}(grad, ug, num) \colondash  \ \ 
 \emph{Univ}(\emph{\hole}_1, \emph{\hole}_2, v_1), \\ \emph{Admit}(v_1, \emph{\hole}_3, \emph{\hole}_4), \ \ 
 \emph{Univ}(\emph{\hole}_5, \emph{\hole}_6, \_), \ \ \emph{Univ}(\emph{\hole}_7, \emph{\hole}_8, \_).
\end{array}
\]
\vspace{-5pt}
\[
\begin{array}{c}
\emph{\hole}_1, \emph{\hole}_3, \emph{\hole}_5, \emph{\hole}_7  \in  \set{id_1, id_2, id_3, uid_1} \  \ \emph{\hole}_4  \in  \set{num, count_1}\\
\emph{\hole}_2, \emph{\hole}_6, \emph{\hole}_8  \in  \set{grad, ug, name_1, name_2, name_3} \\
 \\
\end{array}
\vspace{-8pt}
\]
Suppose the variable for $\emph{\hole}_i$ is $x_i$ and the assignment $\model$ is:
\vspace{-2pt}
\[
\begin{array}{rl}
& x_1 = \emph{id}_1 \land x_2 = \emph{grad} \land x_3 = \emph{id}_1 \land x_4 = \emph{num} \\
\land & x_5 = \emph{id}_1 \land x_6 = \emph{ug} \land x_7 = \emph{id}_2 \land x_8 = \emph{name}_1 \\
\end{array}
\]
\vspace{-2pt}
Since $grad$, $ug$, and $num$ occur in the head,  $Generalize(\model)$ yields the following formula:
\vspace{-5pt}
\begin{equation} \label{eq:generalize1}
\begin{array}{c}
x_2 = grad \land x_4 = num \land x_6 = ug \\
\land x_1 \neq x_2 \land x_1 = x_3 \land x_1 \neq x_4 \land x_1 = x_5 \\
\land x_1 \neq x_6 \land x_1 \neq x_7 \land x_1 \neq x_8 \land \cdots \land x_7 \neq x_8 \\
\end{array}
\end{equation}
\end{example}

The other key insight underlying our sketch completion algorithm is that we can achieve even greater generalization power using the concept of \emph{minimal distinguishing projections (MDP)}, defined as follows:

\vspace{-3pt}
\begin{definition}{\bf (MDP)}
We say that a set of attributes $A$ is a minimal distinguishing projection for 
 Datalog program $\prog$ and input-output example $(\ein, \eout)$ if (1) $\Pi_A(\eout) \neq \Pi_A(\prog(\ein))$, and (2) for any $A' \subset A$, we have $\Pi_{A'}(\eout) = \Pi_{A'}(\prog(\ein))$.
\end{definition}
\vspace{-3pt}

In other words, the first condition ensures that, by just looking at attributes $A$, we can tell that program $\prog$ does not satisfy the examples. On the other hand, the second condition ensures that $A$ is minimal.

To see why minimal distinguishing projections are useful for pruning a larger set of programs, recall that our $\emph{Generalize}(\model)$ function from earlier retains a variable assignment $x \mapsto v$ if $v$ corresponds to a head variable. However, if $v$ does not correspond to an attribute in the MDP, then we will still obtain an incorrect program even if we rename $x$ to something else; thus \emph{Generalize} can drop the assignments to head variables that are not in the MDP.
Thus, given an MDP $\mdp$, we can obtain an improved generalization procedure $\emph{Generalize}(\model, \mdp)$ by using the following $\alpha(x, \model, \mdp)$ function instead of $\alpha(x, \model)$ from earlier:
\[
\alpha(x, \model, \mdp) = \left \{ 
\begin{array}{ll}
x=\model(x)  &  \text{if} \ \model(x) \in \mdp \\
\bigwedge_{x_j \in \emph{dom}(\model)} x \star x_j
& \text{otherwise}
\end{array}
\right .
\]

Because not all head variables  correspond to an MDP attribute, performing generalization  this way allows us to obtain a better blocking clause that rules out many more Datalog programs in one iteration. 
\vspace{-5pt}
\begin{example} \label{ex:mdp}
Consider the same sketch and assignment $\model$ from Example~\ref{ex:generalize}, but now suppose we are given an MDP $\mdp = \set{num}$. Then the function $Generalize(\model, \mdp)$ yields the following more general formula:
\vspace{-5pt}
\begin{equation} \label{eq:generalize2}
\begin{array}{c}
x_4 = num \land x_1 \neq x_2 \land x_1 = x_3 \land x_1 \neq x_4 \land x_1 = x_5 \\
\land x_1 \neq x_6 \land x_1 \neq x_7 \land x_1 \neq x_8 \land \cdots \land x_7 \neq x_8 \\
\end{array}
\vspace{-5pt}
\end{equation}
Note that  (\ref{eq:generalize2}) is more general (i.e., weaker) than  (\ref{eq:generalize1})  because it drops the constraints $x_2 = grad$ and $x_6 = ug$. Therefore, the negation of (\ref{eq:generalize2}) is a better blocking clause than the negation of (\ref{eq:generalize1}), since it rules out more programs in one step.
\end{example}
\vspace{-5pt}

\begin{figure}[!t]
\vspace{-5pt}
\begin{algorithm}[H]
\caption{Analyzing outputs to prune search space}
\label{algo:analyze}
\begin{algorithmic}[1]
\Procedure{\textsc{Analyze}}{$\model, \eout', \eout$}
\vspace{2pt}
\small
\Statex \textbf{Input:} Model $\model$, actual output $\eout'$, expected output $\eout$
\Statex \textbf{Output:} Blocking clause $\phi$
\vspace{2pt}

\State $\phi \leftarrow \emph{true}$;
\State $\mdpSet \leftarrow \textsc{MDPSet}(\eout', \eout)$;
\For{\textbf{each} $\mdp \in \mdpSet$}
    \State $\psi \leftarrow \emph{true}$;
    \For{\textbf{each} $(x_i, x_j) \in \dom(\model) \times \dom(\model)$}
        \If{$\model(x_i) = \model(x_j)$} $\psi \leftarrow \psi \land x_i = x_j$;
        \Else\ $\psi \leftarrow \psi \land x_i \neq x_j$; \EndIf
    \EndFor
    \For{\textbf{each} $x_i \in \dom(\model)$}
        \If{$\model(x_i) \in \mdp$} $\psi \leftarrow \psi \land x_i = \model(x_i)$; \EndIf
    \EndFor
    \State $\phi \leftarrow \phi \land \neg \psi$;
\EndFor
\State \Return $\phi$;

\EndProcedure
\end{algorithmic}
\end{algorithm}
\vspace{-20pt}
\end{figure}

Based on this discussion, we now explain the full {\sc Analyze} procedure in Algorithm~\ref{algo:analyze}.  This procedure takes as input a model $\sigma$ of the SMT encoding and the actual and expected outputs $\eout', \eout$. Then, at line 3, it invokes the {\sc MDPSet} procedure to obtain a set $\mdpSet$ of minimal distinguishing projections and uses each MDP $\mdp \in \mdpSet$ to generate a  blocking clause as discussed above (lines 6--10).

 The {\sc MDPSet} procedure is shown in Algorithm~\ref{algo:mdpset} and uses a breadth-first search algorithm to compute the set of all minimal distinguishing projections.
Specifically, it  initializes a queue $\worklist$ with singleton projections $\set{a}$ for each attribute $a$ in the output (lines 2 -- 5). Then, it repeatedly dequeues a projection $L$ from $\worklist$ and checks if $L$ is an MDP (lines 6 -- 14). In particular, if $L$ can distinguish outputs $\eout'$ and $\eout$ (line 14) and there is no existing projection $L''$ in the current MDP set $\mdpSet$ such that $L'' \subseteq L$, then $L$ is an MDP. If $L$ cannot distinguish outputs $\eout'$ and $\eout$ (line 8), we enqueue all of its extensions $L'$ with one more attribute than $L$ and move on to the next projection in queue $\worklist$.
\vspace{-5pt}
\begin{example}
Let us continue with Example~\ref{ex:mdp} to illustrate how to prune incorrect Datalog programs using multiple MDPs. Suppose we obtain the MDP set $\mdpSet = \set{\mdp_1, \mdp_2}$, where $\mdp_1 = \set{num}$ and $\mdp_2 = \set{grad, ug}$. In addition to $Generalize(\model, \mdp_1)$ (see formula (\ref{eq:generalize2}) of Example~\ref{ex:mdp}), we also compute $Generalize(\model, \mdp_2)$ as:
\vspace{-4pt}
\[
\begin{array}{c}
x_2 = grad \land x_6 = ug \land x_1 \neq x_2 \land x_1 = x_3 \land x_1 \neq x_4 \\
\land x_1 = x_5 \land x_1 \neq x_6 \land x_1 \neq x_7 \land x_1 \neq x_8 \land \cdots \land x_7 \neq x_8 \\
\end{array}
\vspace{-2pt}
\]
By adding  both blocking clauses $\neg Generalize(\model, \mdp_1)$  as well as $\neg Generalize(\model, \mdp_2)$, we can prune even more incorrect Datalog programs.
\end{example}
\begin{figure}[!t]
\vspace{-5pt}
\begin{algorithm}[H]
\caption{Computing a set of MDPs}
\label{algo:mdpset}
\begin{algorithmic}[1]
\Procedure{\textsc{MDPSet}}{$\eout', \eout$}
\vspace{2pt}
\small
\Statex \textbf{Input:} Actual output $\eout'$, expected output $\eout$
\Statex \textbf{Output:} A set of minimal distinguishing projections $\mdpSet$
\vspace{2pt}

\State $\mdpSet \leftarrow \emptyset$; \ $\visited \leftarrow \emptyset$;
\State $\worklist \leftarrow \textsf{EmptyQueue}()$;
\For{\textbf{each} $\attr \in \textsf{Attributes}(\eout')$}
    \State $\worklist.\textsf{Enqueue}(\set{\attr})$; $\visited \leftarrow \visited \cup \set{\set{\attr}}$;
\EndFor
\While{$\neg \worklist.\textsf{IsEmpty}()$}
    \State $L \leftarrow \worklist.\textsf{Dequeue}()$;
    \If{$\Pi_L(\eout') = \Pi_L(\eout)$}
        \For{\textbf{each} $\attr' \in \textsf{Attributes}(\eout') \setminus L$}
            \State $L' \leftarrow L \cup \set{\attr'}$;
            \If{$L' \not\in \visited$}
                \State $\worklist.\textsf{Enqueue}(L')$;
                \State $\visited \leftarrow \visited \cup \set{L'}$;
            \EndIf
        \EndFor
    \ElsIf{$\nexists L'' \in \mdpSet.~ L'' \subseteq L$} $\mdpSet \leftarrow \mdpSet \cup \set{L}$;
    \EndIf
\EndWhile
\State \Return $\mdpSet$;

\EndProcedure
\end{algorithmic}
\end{algorithm}
\vspace{-20pt}
\end{figure}

\vspace{-0.2in}
\begin{theorem} \label{thm:prune-safety}
Let $\phi$ be a blocking clause returned by the call to {\sc Analyze} at line 10 of Algorithm~\ref{algo:synthesis}. If $\sigma$ is a model of $\neg \phi$, then $\sigma$ corresponds to an incorrect Datalog program.
\end{theorem}
\begin{proof}
\ifextended{
See Appendix A.
}\else{
See Appendix A of the extend version~\cite{extended-version}.
}\fi
\end{proof}
\vspace{-3pt}
\section{Implementation} \label{sec:impl}

We have implemented the proposed technique as a new tool called \toolname. Internally, \toolname uses the Z3  solver~\cite{z3-tacas08} for answering SMT queries and leverages the Souffle framework~\cite{souffle-cav16} for evaluating Datalog programs. In the remainder of this section, we discuss some extensions  over the synthesis algorithm described in Section~\ref{sec:synthesis}.

\vspace{5pt}
\noindent \emph{\textbf{Interactive mode.}}
In Section~\ref{sec:synthesis}, we presented our technique as returning a \emph{single} program that is consistent with a given  input-output example. However, in this \emph{non-interactive mode}, \toolname does not guarantee the uniqueness of the program consistent with the given example. To address this potential usability issue, \toolname can also be used in a so-called \emph{interactive mode} where \toolname iteratively queries the user for more examples in order to resolve ambiguities. Specifically, when used in this interactive mode, \toolname first checks if there are multiple programs $\prog, \prog'$ that are consistent with the provided examples $(\ein, \eout)$, and, if so, \toolname identifies a small \emph{differentiating input} $\ein'$ such that  $\prog$ and $\prog'$ yield different outputs on $\ein'$. Then, \toolname asks the user to provide the corresponding output for $\ein'$. More details on how we implement this interactive mode  can be found in
\ifextended{
\!\!\!Appendix B.
}\else{
\!\!\!Appendix B of the extended version~\cite{extended-version}.
}\fi

\begin{example}
Suppose the source database contains two relations Employee(name, deptId) and Department(id, deptName), and we want to obtain the  relation WorksIn(name, deptName) by  joining Employee and Department on deptId=id and then applying projection. Suppose the user only provides the input example Employee(Alice, 11), Department(11, CS) and the  output WorksIn(Alice, CS). Now \toolname may return one of the following results:
\begin{enumerate}[itemsep=-3pt]
\vspace{-3pt}
\item[(1)] WorksIn(x, y) :- Employee(x, z), Department(z, y).
\item[(2)] WorksIn(x, y) :- Employee(x, z), Department(w, y).
\vspace{-3pt}
\end{enumerate}
Note that both Datalog programs are consistent with the input-output example, but only program (1) is the transformation the user wants. Since the program returned by \toolname depends on the model sampled by the SMT solver, it is possible that \toolname returns the incorrect solution (2) instead of the desired program (1).
Using \toolname in the interactive mode solves this problem. In this mode, \toolname searches for an input that distinguishes the two programs shown above. In this case, such a distinguishing input is Employee(Alice, 11), Employee(Bob, 12), Department(11, CS), Department(12, EE), and \toolname asks the user to provide the corresponding output.  Now, if the user provides the output WorksIn(Alice, CS), WorksIn(Bob, EE), only program (1) will be consistent and \toolname successfully eliminates the initial ambiguity.
\end{example}

\noindent \emph{\textbf{Filtering operation.}}
While the synthesis algorithm described in Section~\ref{sec:synthesis} does not support data filtering during  migration, \toolname allows the target database instance to contain a subset of the data in the source instance. However, the filtering operations supported by \toolname are restricted to predicates that can be expressed as a conjunction of equalities. To see how \toolname supports such filtering operations, observe that if an extensional relation $R$ uses a constant $c$ as the argument of attribute $a_i$, this is the same as filtering out  tuples where the corresponding value is not $c$.
Based on this observation, \toolname allows program sketches where the domain of a hole can include constants in addition to variables. These constants are drawn from values in the output example, and the sketch completion algorithm performs enumerative search over these constants.

\vspace{5pt}
\noindent \emph{\textbf{Database instance construction.}}
\toolname builds the target database instance from the output facts of the synthesized Datalog program as described in Section~\ref{sec:prelim-migration}. However, \toolname performs one optimization to make large-scale data migration practical: We leverage MongoDB~\cite{mongodb-web} to build indices on attributes that connect records to their parents. This strategy allows \toolname to quickly look up the children of a given record and  makes the construction of the target database  more efficient.

\section{Evaluation} \label{sec:eval}

To evaluate \toolname, we perform experiments that are designed to answer the following research questions:
\begin{itemize}\itemsep0em
    \item[\textbf{RQ1}] Can \toolname successfully migrate real-world data sets given a representative set of records, and how good are the synthesized programs?
    \item[\textbf{RQ2}] How sensitive is the synthesizer to the number and quality of examples?
    \item[\textbf{RQ3}] How helpful is \toolname for users in practice, and how do users choose between multiple correct answers? 
    \item[\textbf{RQ4}] Is the proposed sketch completion algorithm significantly more efficient than a simpler baseline?
    \item[\textbf{RQ5}] How does the proposed synthesis technique compare against prior techniques? 
\end{itemize}
\begin{table}[!t]
\centering
\small
\begin{tabular}{|c|c|l|}
\hline
\textbf{Name} & \textbf{Size} & \textbf{Description} \\
\hline
Yelp & 4.7GB & Business and reviews from Yelp \\
\hline
IMDB & 6.3GB & Movie and crew info from IMDB \\
\hline
Mondial & 3.7MB & Geography information \\
\hline
DBLP & 2.0GB & Publication records from DBLP \\
\hline
MLB & 0.9GB & Pitch data of Major League Baseball \\
\hline
Airbnb & 0.4GB & Berlin Airbnb data \\
\hline
Patent & 1.7GB & Patent Litigation Data 1963-2015 \\
\hline
Bike & 2.7GB & Bike trip data in Bay Area \\
\hline
Tencent & 1.0GB & User followers in Tencent Weibo \\
\hline
Retina & 0.1GB & Biological info of mouse retina \\
\hline
Movie & 0.1GB & Movie ratings from MovieLens \\
\hline
Soccer & 0.2GB & Transfer info of soccer players \\
\hline
\end{tabular}
\vspace{5pt}
\caption{Datasets used in the evaluation.}
\label{tab:datasets}
\vspace{-15pt}
\end{table}

\noindent \emph{\textbf{Benchmarks.}}
To answer these research questions, we collected $12$ real-world database instances (see Table~\ref{tab:datasets} for details) and created $28$ benchmarks in total.
Specifically,  four of these datasets (namely Yelp, IMDB, Mondial, and DBLP) are taken from prior work~\cite{mitra-vldb18}, and the remaining eight are taken from open dataset websites such as Kaggle~\cite{kaggle-web}. For the document-to-relational transformations, we used exactly the same benchmarks as prior work~\cite{mitra-vldb18}. For the remaining cases (e.g., document-to-graph or graph-to-relational), we used the source schemas in the original dataset but created a suitable target schema ourselves.
As summarized in Table~\ref{tab:benchmarks}, our $28$ benchmarks collectively cover a broad range of migration scenarios between different types of databases.\footnote{Schemas for all benchmarks are available at \url{https://bit.ly/schemas-dynamite}.}

\begin{table}
\centering
\small
\begin{tabular}{|c|c|c|c|c|c|c|}
\hline
\multirow{2}{*}{\textbf{\!\!Benchmark\!\!}} &
\multicolumn{3}{c|}{\textbf{Source Schema}} &
\multicolumn{3}{c|}{\textbf{Target Schema}} \\
\cline{2-7}
& \textbf{\!\!Type\!\!} & \textbf{\!\!\#Recs\!\!} & \textbf{\!\!\#Attrs\!\!} & \textbf{\!\!Type\!\!} & \textbf{\!\!\#Recs\!\!} & \textbf{\!\!\#Attrs\!\!} \\
\hline
Yelp-1 & D & 11 & 58 & R & 8 & 32 \\
\hline
IMDB-1 & D & 12 & 21 & R & 9 & 26 \\
\hline
DBLP-1 & D & 37 & 42 & R & 9 & 35 \\
\hline
Mondial-1 & D & 37 & 113 & R & 25 & 110 \\
\hline
MLB-1 & R & 5 & 83 & D & 7 & 85 \\
\hline
Airbnb-1 & R & 4 & 30 & D & 6 & 24 \\
\hline
Patent-1 & R & 5 & 49 & D & 7 & 50 \\
\hline
Bike-1 & R & 4 & 48 & D & 7 & 47 \\
\hline
Tencent-1 & G & 2 & 8 & R & 1 & 3 \\
\hline
Retina-1 & G & 2 & 17 & R & 2 & 13 \\
\hline
Movie-1 & G & 5 & 18 & R & 5 & 21 \\
\hline
Soccer-1 & G & 10 & 30 & R & 7 & 21 \\
\hline
Tencent-2 & G & 2 & 8 & D & 1 & 3 \\
\hline
Retina-2 & G & 2 & 17 & D & 2 & 15 \\
\hline
Movie-2 & G & 5 & 18 & D & 4 & 14 \\
\hline
Soccer-2 & G & 10 & 30 & D & 7 & 23 \\
\hline
Yelp-2 & D & 11 & 58 & G & 4 & 31 \\
\hline
IMDB-2 & D & 12 & 21 & G & 11 & 19 \\
\hline
DBLP-2 & D & 37 & 42 & G & 17 & 28 \\
\hline
Mondial-2 & D & 37 & 113 & G & 27 & 78 \\
\hline
MLB-2 & R & 5 & 83 & G & 12 & 90 \\
\hline
Airbnb-2 & R & 4 & 30 & G & 7 & 32 \\
\hline
Patent-2 & R & 5 & 49 & G & 8 & 49 \\
\hline
Bike-2 & R & 4 & 48 & G & 6 & 52 \\
\hline
MLB-3 & R & 5 & 83 & R & 4 & 75 \\
\hline
Airbnb-3 & R & 4 & 30 & R & 7 & 33 \\
\hline
Patent-3 & R & 5 & 49 & R & 8 & 52 \\
\hline
Bike-3 & R & 4 & 48 & R & 5 & 52 \\
\hline
\hline
\textbf{Average} & - & \textbf{10.2} & \textbf{44.4} & - & \textbf{8.0} & \textbf{39.8} \\
\hline
\end{tabular}
\vspace{10pt}
\caption{Statistics of benchmarks. ``R'' stands for relational, ``D'' stands for document, and ``G'' stands for graph.}
\label{tab:benchmarks}
\vspace{-10pt}
\end{table}

\vspace{5pt}
\noindent \emph{\textbf{Experimental setup.}}
All experiments are conducted on a machine with Intel Xeon(R) E5-1620 v3 quad-core CPU and 32GB of physical memory, running the Ubuntu 18.04 OS.

\subsection{Synthesis Results in Non-Interactive Mode} \label{sec:exp-main}

In this section, we evaluate RQ1 by using \toolname to migrate  the datasets from Table~\ref{tab:datasets} for the source and target schemas from Table~\ref{tab:benchmarks}. To perform this experiment, we first constructed a representative set of input-output examples for each record in the source and target schemas. As shown in Table~\ref{tab:results}, across all benchmarks, the average number of records in the input (resp. output) example is $2.6$ (resp. $2.2$). Given these examples, we then used \toolname to synthesize a migration script consistent with the given examples and ran it on the real-world datasets from Table~\ref{tab:datasets}.\footnote{All input-output examples and synthesized programs are available at \url{https://bit.ly/benchmarks-dynamite}.}
We now highlight the key take-away lessons from this experiment  whose results are summarized in Table~\ref{tab:results}.

\vspace{3pt}
\noindent \emph{\textbf{Synthesis time.}}
Even though the search space of possible Datalog programs is very large ($5.1 \times 10^{39}$ on average), \toolname can find a Datalog program consistent with the examples in an average of $7.3$ seconds, with maximum synthesis time being $87.9$ seconds.

\begin{table*}
\centering
\small
\begin{tabular}{|c|c|c|c|c|c|c|c|c|c|}
\hline
\multirow{2}{*}{\textbf{\!\!Benchmark\!\!}} &
\multicolumn{2}{c|}{\!\textbf{Avg \# Examples}\!} &
\textbf{Search} &
\!\!\textbf{Synthesis}\!\! &
\multirow{2}{*}{\textbf{\# Rules}} &
\textbf{\# Preds} &
\textbf{\# Optim} &
\textbf{Dist to} &
\!\!\textbf{Migration}\!\! \\
\cline{2-3}
& \textbf{Source} & \textbf{Target} & \textbf{Space} & \textbf{Time (s)} & & \textbf{per Rule} & \textbf{Rules} & \textbf{Optim} & \textbf{Time (s)} \\
\hline
Yelp-1 & 4.7 & 3.9 & $4.8 \times 10^{120}$ & 6.0 & 8 & 1.8 & 7 & 0.38 & 328 \\
\hline
IMDB-1 & 6.0 & 2.7 & $1.5 \times 10^{20}$ & 2.7 & 9 & 3.6 & 5 & 1.22 & 1153 \\
\hline
DBLP-1 & 1.5 & 2.6 & $1.1 \times 10^{14}$ & 0.8 & 9 & 6.4 & 0 & 2.44 & 1060 \\
\hline
Mondial-1 & 1.2 & 2.8 & $2.2 \times 10^{88}$ & 2.5 & 25 & 3.3 & 17 & 1.40 & 5 \\
\hline
MLB-1 & 2.0 & 1.4 & $9.1 \times 10^{81}$ & 13.0 & 7 & 3.9 & 2 & 1.71 & 1020 \\
\hline
Airbnb-1 & 4.0 & 2.5 & $1.7 \times 10^{38}$ & 2.0 & 6 & 2.7 & 4 & 1.33 & 286 \\
\hline
Patent-1 & 2.6 & 2.3 & $1.4 \times 10^{49}$ & 3.0 & 7 & 2.4 & 5 & 1.14 & 553 \\
\hline
Bike-1 & 2.3 & 2.0 & $3.1 \times 10^{47}$ & 2.0 & 7 & 2.0 & 5 & 0.71 & 2601 \\
\hline
Tencent-1 & 1.5 & 1.0 & $1.3 \times 10^{12}$ & 0.2 & 1 & 4.0 & 0 & 3.00 & 65 \\
\hline
Retina-1 & 1.5 & 1.5 & $3.1 \times 10^{19}$ & 0.8 & 2 & 2.0 & 2 & 0.00 & 9 \\
\hline
Movie-1 & 3.6 & 2.2 & $5.2 \times 10^{11}$ & 2.9 & 5 & 2.8 & 3 & 1.00 & 1062 \\
\hline
Soccer-1 & 1.9 & 2.0 & $2.9 \times 10^{11}$ & 0.5 & 7 & 1.0 & 7 & 0.00 & 15 \\
\hline
Tencent-2 & 1.5 & 1.0 & $1.3 \times 10^{12}$ & 0.2 & 1 & 4.0 & 0 & 3.00 & 160 \\
\hline
Retina-2 & 2.0 & 2.0 & $3.3 \times 10^{19}$ & 4.0 & 2 & 2.5 & 1 & 0.50 & 22 \\
\hline
Movie-2 & 2.4 & 2.3 & $1.0 \times 10^{18}$ & 22.7 & 4 & 7.0 & 0 & 4.00 & 40 \\
\hline
Soccer-2 & 2.5 & 2.1 & $6.9 \times 10^{22}$ & 87.9 & 7 & 4.4 & 4 & 1.71 & 311 \\
\hline
Yelp-2 & 4.5 & 1.8 & $2.9 \times 10^{73}$ & 0.5 & 4 & 1.0 & 4 & 0.00 & 1160 \\
\hline
IMDB-2 & 2.4 & 2.5 & $2.3 \times 10^{11}$ & 1.1 & 11 & 3.1 & 5 & 1.27 & 3409 \\
\hline
DBLP-2 & 2.1 & 2.1 & $1.2 \times 10^{4}$ & 3.6 & 17 & 1.8 & 16 & 0.06 & 1585 \\
\hline
Mondial-2 & 1.0 & 2.1 & $8.2 \times 10^{24}$ & 30.8 & 27 & 1.9 & 26 & 0.04 & 7 \\
\hline
MLB-2 & 2.2 & 1.9 & $3.3 \times 10^{84}$ & 2.6 & 12 & 1.3 & 10 & 0.25 & 785 \\
\hline
Airbnb-2 & 2.8 & 2.7 & $1.4 \times 10^{28}$ & 0.9 & 7 & 1.3 & 7 & 0.00 & 664 \\
\hline
Patent-2 & 2.0 & 2.1 & $3.9 \times 10^{51}$ & 1.0 & 8 & 1.4 & 6 & 0.38 & 786 \\
\hline
Bike-2 & 2.3 & 2.5 & $7.3 \times 10^{47}$ & 0.4 & 6 & 1.8 & 4 & 0.83 & 3346 \\
\hline
MLB-3 & 2.2 & 1.3 & $9.1 \times 10^{81}$ & 3.3 & 4 & 2.3 & 3 & 0.50 & 145 \\
\hline
Airbnb-3 & 2.5 & 2.6 & $3.3 \times 10^{28}$ & 0.5 & 7 & 1.1 & 7 & 0.00 & 57 \\
\hline
Patent-3 & 2.8 & 2.3 & $1.3 \times 10^{40}$ & 3.9 & 8 & 1.6 & 7 & 0.38 & 122 \\
\hline
Bike-3 & 4.3 & 2.2 & $7.3 \times 10^{47}$ & 4.1 & 5 & 1.8 & 4 & 0.20 & 519 \\
\hline
\hline
\textbf{Average} & \textbf{2.6} & \textbf{2.2} & $\mathbf{5.1 \times 10^{39}}$ & \textbf{7.3} & \textbf{8.0} & \textbf{2.5} & \textbf{5.8} & \textbf{0.79} & \textbf{760} \\
\hline
\end{tabular}
\vspace{5pt}
\caption{Main results. Average search space size is calculated by geometric mean; all other averages are arithmetic mean.}
\label{tab:results}
\vspace{-10pt}
\end{table*}

\vspace{3pt}
\noindent \emph{\textbf{Statistics about synthesized programs.}}
As shown in Table~\ref{tab:results}, the average number of rules in the synthesized Datalog program is $8.0$, and each rule contains an average of $2.5$ predicates in the rule body (after simplification).

\vspace{3pt}
\noindent \emph{\textbf{Quality of synthesized programs.}}
To evaluate the quality of the synthesized programs, we compared the synthesized Datalog programs against manually written ones (which we believe to be optimal). As shown in the column labeled ``\# Optim Rules'' in Table~\ref{tab:results}, on average, $5.8$ out of the $8$  Datalog rules ($72.5\%$) are \emph{syntactically identical} to the manually-written ones. In cases where the synthesized rule differs from the manually-written one, we observed that the synthesized program contains redundant  body predicates. In particular, if we quantify the distance between the two programs in terms of additional predicates, we found that the synthesized rules contain an average of $0.79$ extra predicates (shown in column labeled ``Dist to Optim''). However, note that, even in cases where the synthesized rule differs syntactically from the manually-written rule, we confirmed that the synthesized and manual rules produce \emph{the exact same output} for the given input relations in \emph{all cases}.

\vspace{3pt}
\noindent \emph{\textbf{Migration time and results.}}
For all $28$ benchmarks,  we confirmed that \toolname is able to produce the intended target database instance. As reported in the column labeled ``Migration time'', the average time taken by \toolname to convert the source instance to the target one is $12.7$ minutes for database instances containing $1.7$ GB of data on average.

\subsection{Sensitivity to Examples}

\begin{figure}[!t]
\centering
    \begin{subfigure}[b]{0.23\textwidth}
        \centering
        \begin{tikzpicture}[scale=0.7]
        \begin{axis}[
            width=2.15in,
            x label style={below=-5pt, font=\small},
            y label style={below=8pt, font=\small},
            axis y line*=left,
            xmin=0,
            xtick={0,2,...,11},
            ymin=0,
            ymax=0.7,
        	xlabel={\# Examples},
        	ylabel={Synthesis Time (s)},
        ]
        \addplot[color=blue,mark=*,width=2pt]
            table[x=a, y=b] {yelp-time.data};
        \label{legend:bench1-time}
        \end{axis}

        \begin{axis}[
            width=2.15in,
            legend cell align=left,
            legend style={at={(0.5,1.17)}, font=\small, legend columns=-1, anchor=north},
            y label style={below=160pt, font=\small},
            legend image post style={scale=0.5},
            hide x axis,
            axis y line*=right,
            xmin=0,
            xtick={0,2,...,11},
            ymin=0,
            ymax=110,
        	ylabel={Success Rate (\%)},
        ]
        \addlegendimage{/pgfplots/refstyle=legend:bench1-time}\addlegendentry{Time(s)}
        \addplot[color=red,mark=square*,width=2pt]
            table[x=a, y=b] {yelp-rate.data};
        \addlegendentry{Succ Rate(\%)}
        \end{axis}
    \end{tikzpicture}
    \vspace{-15pt}
    \caption{Yelp-1}
    \vspace{-3pt}
    \end{subfigure}
    \begin{subfigure}[b]{0.23\textwidth}
        \centering
        \begin{tikzpicture}[scale=0.7]
        \begin{axis}[
            width=2.15in,
            x label style={below=-5pt, font=\small},
            y label style={below=8pt, font=\small},
            axis y line*=left,
            xmin=0,
            xtick={0,2,...,11},
            ymin=0,
            ymax=25,
        	xlabel={\# Examples},
        	ylabel={Synthesis Time (s)},
        ]
        \addplot[color=blue,mark=*,width=2pt]
            table[x=a, y=b] {imdb-time.data};
        \label{legend:bench2-time}
        \end{axis}

        \begin{axis}[
            width=2.15in,
            legend cell align=left,
            legend style={at={(0.5,1.17)}, font=\small, legend columns=-1, anchor=north},
            y label style={below=160pt, font=\small},
            legend image post style={scale=0.5},
            hide x axis,
            axis y line*=right,
            xmin=0,
            xtick={0,2,...,11},
            ymin=0,
            ymax=110,
        	ylabel={Success Rate (\%)},
        ]
        \addlegendimage{/pgfplots/refstyle=legend:bench2-time}\addlegendentry{Time(s)}
        \addplot[color=red,mark=square*,width=2pt]
            table[x=a, y=b] {imdb-rate.data};
        \addlegendentry{Succ Rate(\%)}
        \end{axis}
    \end{tikzpicture}
    \vspace{-15pt}
    \caption{IMDB-1}
    \vspace{-3pt}
    \end{subfigure}
\\ \vspace{10pt}
    \begin{subfigure}[b]{0.23\textwidth}
        \centering
        \begin{tikzpicture}[scale=0.7]
        \begin{axis}[
            width=2.15in,
            x label style={below=-5pt, font=\small},
            y label style={below=8pt, font=\small},
            axis y line*=left,
            xmin=0,
            xtick={0,2,...,11},
            ymin=0,
            ymax=0.9,
        	xlabel={\# Examples},
        	ylabel={Synthesis Time (s)},
        ]
        \addplot[color=blue,mark=*,width=2pt]
            table[x=a, y=b] {dblp-time.data};
        \label{legend:bench3-time}
        \end{axis}

        \begin{axis}[
            width=2.15in,
            legend cell align=left,
            legend style={at={(0.5,1.17)}, font=\small, legend columns=-1, anchor=north},
            y label style={below=160pt, font=\small},
            legend image post style={scale=0.5},
            hide x axis,
            axis y line*=right,
            xmin=0,
            xtick={0,2,...,11},
            ymin=0,
            ymax=110,
        	ylabel={Success Rate (\%)},
        ]
        \addlegendimage{/pgfplots/refstyle=legend:bench3-time}\addlegendentry{Time(s)}
        \addplot[color=red,mark=square*,width=2pt]
            table[x=a, y=b] {dblp-rate.data};
        \addlegendentry{Succ Rate(\%)}
        \end{axis}
    \end{tikzpicture}
    \vspace{-15pt}
    \caption{DBLP-1}
    \vspace{-3pt}
    \end{subfigure}
    \begin{subfigure}[b]{0.23\textwidth}
        \centering
        \begin{tikzpicture}[scale=0.7]
        \begin{axis}[
            width=2.15in,
            x label style={below=-5pt, font=\small},
            y label style={below=8pt, font=\small},
            axis y line*=left,
            xmin=0,
            xtick={0,2,...,11},
            ymin=0,
            ymax=4,
        	xlabel={\# Examples},
        	ylabel={Synthesis Time (s)},
        ]
        \addplot[color=blue,mark=*,width=2pt]
            table[x=a, y=b] {mondial-time.data};
        \label{legend:bench4-time}
        \end{axis}

        \begin{axis}[
            width=2.15in,
            legend cell align=left,
            legend style={at={(0.5,1.17)}, font=\small, legend columns=-1, anchor=north},
            y label style={below=160pt, font=\small},
            legend image post style={scale=0.5},
            hide x axis,
            axis y line*=right,
            xmin=0,
            xtick={0,2,...,11},
            ymin=0,
            ymax=110,
        	ylabel={Success Rate (\%)},
        ]
        \addlegendimage{/pgfplots/refstyle=legend:bench4-time}\addlegendentry{Time(s)}
        \addplot[color=red,mark=square*,width=2pt]
            table[x=a, y=b] {mondial-rate.data};
        \addlegendentry{Succ Rate(\%)}
        \end{axis}
    \end{tikzpicture}
    \vspace{-15pt}
    \caption{Mondial-1}
    \vspace{-3pt}
    \end{subfigure}
\vspace{2pt}
\caption{Sensitivity analysis.}
\label{fig:sensitivity}
\vspace{-17pt}
\end{figure}
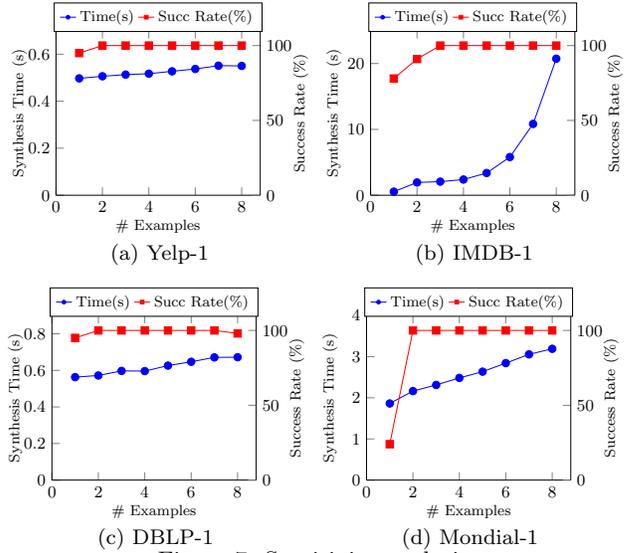

To answer  RQ2, we perform an experiment that measures the sensitivity of \toolname to the number and quality of records in the provided input-output examples.  To perform this experiment, we first fix the number $r$ of records in the input example. Then, we randomly generate $100$ examples of size $r$ and obtain the output example by running the ``golden'' program (written manually) on the randomly generated input example. Then, for each size $r \in [1,8]$, we measure average running time across all $100$ examples as well as the percentage of examples for which \toolname synthesizes the correct program within 10 minutes.

The results of this experiment are summarized in Figure~\ref{fig:sensitivity} for four representative benchmarks (the remaining $24$ are provided in
\ifextended{
\!Appendix C).
}\else{
\!Appendix C of the extend version~\cite{extended-version}).
}\fi
Here, the $x$-axis shows the number of records $r$ in the input example, and the $y$-axis shows both (a) the average running time in seconds for each $r$ (the blue line with circles) and (b) the $\%$ of correctly synthesized programs given $r$ randomly-generated records (the red line with squares).

Overall, this experiment shows that \toolname is not particularly sensitive to the number and quality of examples for $26$ out of $28$ benchmarks: it can synthesize the correct Datalog program in over $90\%$ of the cases using $2-3$ \emph{randomly-generated} examples. Furthermore,  synthesis time grows roughly linearly for $24$ out of $28$ benchmarks. For $2$ of the remaining benchmarks (namely, IMDB-1 and Movie-2), synthesis time seems to grow exponentially in example size; however, since \toolname can already achieve a success rate over $90\%$  with just 2-3 examples, this growth in running time is not a major concern. Finally, for the last $2$ benchmarks (namely, Retina-2 and Soccer-2), \toolname does not seem to scale beyond example size of $3$.  For these benchmarks, \toolname seems to generate complicated intermediate programs with complex join structure, which causes the resulting output to be very large and causes MDP analysis to become very time-consuming.  However, this behavior (which is triggered by randomly generated inputs) can be prevented by choosing more representative examples that allow \toolname to generate better sketches.

\subsection{User Study}

To answer question RQ3, we conduct a small-scale user study to evaluate whether \toolname is helpful to users in practice. To conduct this user study, we recruited 10  participants (all of them graduate computer science students with advanced programming skills) and asked them to solve two of our benchmarks from Table~\ref{tab:benchmarks}, namely Tencent-1 and Retina-1, with and without using \toolname.  In order to avoid any potential bias, we randomly assigned  users to solve each benchmark either using \toolname or without. For the setting where users were not allowed to use \toolname, they were, however, permitted to use any programming language and library of their choice, and they were also allowed to consult search engines and on-line forums. In the setting where the participants did use \toolname, we instructed them to use the tool in interactive mode (recall Section~\ref{sec:impl}). Overall, exactly $5$ randomly chosen participants solved Tencent-1 and Retina-1 using \toolname, and $5$ users solved each benchmark without \toolname.

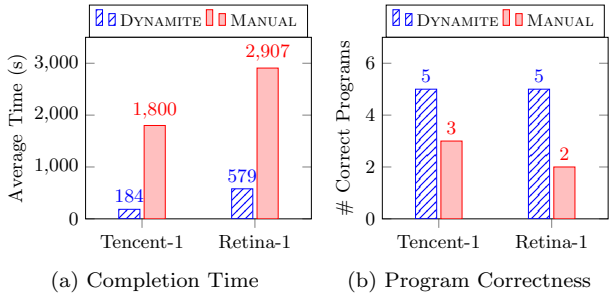
\begin{figure}[!t]
\centering
    \begin{subfigure}[b]{0.23\textwidth}
        \centering
        \begin{tikzpicture}[scale=0.8]
        \begin{axis}[
            width=2.1in,
            ybar,
            legend cell align=left,
            legend entries={\toolname, \textsc{Manual}},
            legend style={at={(0.5,1.18)}, font=\small, legend columns=-1, anchor=north},
            x label style={below=0pt},
            y label style={below=-5pt},
            ylabel={Average Time (s)},
            symbolic x coords={Tencent-1,Retina-1},
            enlarge x limits=0.5,
            xtick=data,
            ymin=0,
            ymax=3500,
            nodes near coords,
            nodes near coords align={vertical},
        ]
        \addplot[color=blue,pattern=north east lines,pattern color=blue]
            table [x=a, y=b] {dynamite-manual.data};
        \addplot[color=red,fill=pink]
            table [x=a, y=b] {manual.data};
        \end{axis}
        \end{tikzpicture}
        \caption{Completion Time}
    \end{subfigure}
    \begin{subfigure}[b]{0.23\textwidth}
        \centering
        \begin{tikzpicture}[scale=0.8]
        \begin{axis}[
            width=2.1in,
            ybar,
            legend cell align=left,
            legend entries={\toolname, \textsc{Manual}},
            legend style={at={(0.5,1.18)}, font=\small, legend columns=-1, anchor=north},
            x label style={below=0pt},
            y label style={below=10pt},
            ylabel={\# Correct Programs},
            symbolic x coords={Tencent-1,Retina-1},
            enlarge x limits=0.5,
            xtick=data,
            ymin=0,
            ymax=7,
            nodes near coords,
            nodes near coords align={vertical},
        ]
        \addplot[color=blue,pattern=north east lines,pattern color=blue]
            table [x=a, y=b] {dynamite-count.data};
        \addplot[color=red,fill=pink]
            table [x=a, y=b] {manual-count.data};
        \end{axis}
        \end{tikzpicture}
        \caption{Program Correctness}
    \end{subfigure}
\vspace{2pt}
\caption{Results of user study.}
\label{fig:user-study}
\vspace{-7pt}
\end{figure}

The results of this study are provided in Figure~\ref{fig:user-study}.  Specifically, Figure~\ref{fig:user-study}(a) shows the average time to solve the two benchmarks with and without using \toolname. As we can see from this Figure, users are significantly more productive, namely by a factor of $6.2$x on average, when migrating  data with the aid of \toolname. Furthermore, as shown in Figure~\ref{fig:user-study}(b), participants always generate the correct target database instance when using \toolname; however, they fail to do so in $50\%$ of the cases when they write the program manually. Upon further inspection, we found that the manually written programs contain subtle bugs, such as failing to introduce newlines or quotation marks. We believe these results demonstrate that \toolname can aid users, including seasoned programmers, successfully and more efficiently complete real-world data migration tasks.

\subsection{Comparison with Synthesis Baseline} \label{sec:exp-baseline}

To answer   \text{RQ4}, we  compare \toolname against a baseline called \toolvar that uses enumerative search instead of the sketch completion technique described in Section~\ref{sec:sketch-completion}. In particular, \toolvar uses the lazy enumeration algorithm based on SMT, but it does not use the {\sc Analyze} procedure for learning from failed synthesis attempts. Specifically, whenever the SMT solver returns an incorrect assignment $\sigma$, \toolvar just uses $\neg \sigma$ as a blocking clause. Thus, \toolvar essentially enumerates all possible sketch completions  until it finds a Datalog program that satisfies the input-output example.

Figure~\ref{fig:comparison}(a) shows the results of the comparison when using the manually-provided input-output examples from Section~\ref{sec:exp-main}. In particular, we plot the time in seconds that each version takes to solve the first $n$ benchmarks. As shown in Figure~\ref{fig:comparison}(a), \toolname can successfully solve all $28$ benchmarks whereas \toolvar can  only solve $22$ ($78.6\%$) within the one hour time limit.  Furthermore, for the first $22$ benchmarks that can be solved by both versions, \toolname is $9.2$x faster compared to \toolvar ($1.8$ vs $16.5$ seconds).
Hence, this experiment demonstrates the practical advantages of our proposed  sketch completion algorithm compared to a simpler enumerative-search baseline.

\begin{figure}[!t]
\centering
    \begin{subfigure}[b]{0.23\textwidth}
        \centering
        \begin{tikzpicture}[scale = 0.75]
        \begin{axis}[
            width=2.3in,
            legend cell align=left,
            legend entries={\toolname, \toolvar},
            legend style={at={(0,1.0)}, font=\scriptsize, legend columns=1, anchor=north west},
            x label style={below=0pt},
            y label style={below=5pt},
            ymin = -2,
        	xlabel={\# Solved Benchmarks},
        	ylabel={Synthesis Time (s)},
        ]
        \addplot table [color=blue, x=a, y=b] {dynamite.data};
        \addplot table [color=red, x=a, y=b] {enum.data};
        \end{axis}
        \end{tikzpicture}
        \caption{Comparison with enum}
    \end{subfigure}
    \begin{subfigure}[b]{0.23\textwidth}
        \centering
        \begin{tikzpicture}[scale=0.75]
        \begin{axis}[
            width=2.2in,
            ybar,
            bar width=10pt,
            legend cell align=left,
            legend entries={\toolname, \mitra},
            legend style={at={(0.5,1.17)}, legend columns=-1, anchor=north},
            x label style={below=0pt},
            y label style={below=5pt},
            x tick label style = {font=\scriptsize},
            ylabel={Synthesis Time (s)},
            symbolic x coords={Yelp,IMDB,DBLP,Mondial},
            xtick=data,
            enlarge x limits=0.2,
            ymin=0,
            ymax=80,
            nodes near coords,
            nodes near coords align={vertical},
        ]
        \addplot[color=blue,pattern=north east lines,pattern color=blue]
            table [x=a, y=b] {dynamite-mitra.data};
        \addplot[color=red,fill=pink]
            table [x=a, y=b] {mitra.data};
        \end{axis}
        \end{tikzpicture}
        \caption{Comparison with \mitra}
    \end{subfigure}
\vspace{5pt}
\caption{Comparing \toolname to baseline and \mitra.}
\label{fig:comparison}
\vspace{-10pt}
\end{figure}

\subsection{Comparison with Other Tools}\label{sec:mitra-comp}

While there is no existing programming-by-example (PBE) tool that supports the full diversity of source/target schemas handled by \toolname, we compare our approach against two other  tools, namely \mitra and \eirene, in two more specialized data migration scenarios. Specifically, \mitra~\cite{mitra-vldb18} is a PBE tool that automates document-to-relational transformations, whereas \eirene~\cite{eirene-pvldb11} infers relational-to-relational schema mappings from input-output examples.

\vspace{3pt}
\noindent \emph{\textbf{Comparison with Mitra.}} 
Since \mitra uses a domain-specific language that is customized for transforming tree-structured data into a tabular representation,  we compare \toolname against \mitra on the four data migration benchmarks from~\cite{mitra-vldb18} that involve conversion from a document  schema to a relational schema.  The results of this comparison   are summarized in Figure~\ref{fig:comparison}(b), which shows  synthesis time for each tool for all  four benchmarks. In terms of synthesis time, \toolname outperforms \mitra by roughly an order of magnitude:  in particular, \toolname takes an average of  $3$ seconds to solve these benchmarks, whereas \mitra needs  $29.4$ seconds.
Furthermore, \mitra synthesizes $559$ and $780$ lines of JavaScript for Yelp and IMDB, and synthesizes $134$ and $432$ lines of XSLT for DBLP and Mondial. In contrast, \toolname  synthesizes $13$ Datalog rules on average. These statistics suggest that the programs synthesized by \toolname are more easily readable compared to the JavaScript and XSLT programs synthesized by \mitra.
Finally, if we compare \toolname and \mitra in terms of efficiency of the synthesized programs, we observe that  \toolname-generated programs are $1.1$x faster.

\begin{figure}[!t]
\centering
    \begin{subfigure}[b]{0.23\textwidth}
        \centering
        \begin{tikzpicture}[scale=0.8]
        \begin{axis}[
            width=2.3in,
            ybar,
            legend cell align=left,
            legend entries={\toolname, \eirene},
            legend style={at={(0.5,1.16)}, font=\small, legend columns=-1, anchor=north},
            x label style={below=0pt},
            y label style={below=15pt},
            ylabel={Synthesis Time (s)},
            x tick label style={font=\scriptsize},
            symbolic x coords={MLB-3,Airbnb-3,Patent-3,Bike-3},
            enlarge x limits=0.15,
            xtick=data,
            ymin=0,
            ymax=6,
            nodes near coords,
            nodes near coords align={vertical},
        ]
        \addplot[color=blue,pattern=north east lines,pattern color=blue]
            table [x=a, y=b] {dynamite-rel-time.data};
        \addplot[color=red,fill=pink]
            table [x=a, y=b] {eirene-time.data};
        \end{axis}
        \end{tikzpicture}
        \caption{Synthesis Time}
    \end{subfigure}
    \begin{subfigure}[b]{0.23\textwidth}
        \centering
        \begin{tikzpicture}[scale=0.8]
        \begin{axis}[
            width=2.3in,
            ybar,
            legend cell align=left,
            legend entries={\toolname, \eirene},
            legend style={at={(0.5,1.16)}, font=\small, legend columns=-1, anchor=north},
            x label style={below=0pt},
            y label style={below=15pt},
            ylabel={Distance to Optimal},
            x tick label style={font=\scriptsize},
            symbolic x coords={MLB-3,Airbnb-3,Patent-3,Bike-3},
            enlarge x limits=0.15,
            xtick=data,
            ymin=0,
            ymax=2.5,
            nodes near coords,
            nodes near coords align={vertical},
        ]
        \addplot[color=blue,pattern=north east lines,pattern color=blue]
            table [x=a, y=b] {dynamite-rel-quality.data};
        \addplot[color=red,fill=pink]
            table [x=a, y=b] {eirene-quality.data};
        \end{axis}
        \end{tikzpicture}
        \caption{Mapping Quality}
    \end{subfigure}
\vspace{5pt}
\caption{Comparing \toolname against \eirene.}
\label{fig:comp-eirene}
\vspace{-15pt}
\end{figure}
\vspace{3pt}
\noindent \emph{\textbf{Comparison with Eirene.}} 
Since \eirene specializes in inferring relational-to-relational schema mappings, we compare \toolname against \eirene on the four relational-to-relational benchmarks from Section~\ref{sec:exp-main} using the same input-output  examples.
As shown in Figure~\ref{fig:comp-eirene}(a), \toolname is, on average, $1.3$x faster than \eirene in terms of synthesis time.  We also compare \toolname with \eirene in terms of the quality of inferred mappings using the same ``distance from optimal schema mapping metric'' defined in Section~\ref{sec:exp-main}.\footnote{To conduct this measurement, we  manually wrote optimal schema mappings in the formalism used by \eirene.} As shown in Figure~\ref{fig:comp-eirene}(b), the schema mappings synthesized by \toolname are closer to the optimal mappings than those synthesized by \eirene. In particular, \eirene-synthesized rules have $4.5$x more redundant body predicates than the \toolname-synthesized rules. 

\section{Related Work}

\noindent \emph{\textbf{Schema mapping formalisms.}} There are several different formalisms for  expressing schema mappings, including visual representations~\cite{survey-vldbj01, mweaver-sigmod12}, schema modification operators~\cite{prism-vldb08, smo-vldb13}, and declarative constraints~\cite{glav-sigmod11, tgd-pods10, eirene-pvldb11, Fagin-tcs05, Kolaitis-pods05, approx-tods17, skolem-sigmod13}. Some techniques require the user to express the schema mapping visually by drawing arrows between attributes in the source and target schemas~\cite{survey-vldbj01, clio-vldb02}. In contrast, schema modification operators express the schema mapping in a domain-specific language~\cite{prism-vldb08, smo-vldb13}. Another common approach is to express the schema mapping using declarative specifications, such as Global-Local-As-View (GLAV) constraints~\cite{glav-sigmod11, tgd-pods10, eirene-pvldb11, Fagin-tcs05}. Similar to this third category, we express the schema mapping using a declarative, albeit executable, formalism.

\vspace{3pt}
\noindent \emph{\textbf{Schema mapping inference.}} There is also a body of prior work on \emph{automatically inferring}  schema mappings~\cite{clio-vldb00, Yan-sigmod01, An-icde07, clio-essay09, clio-vldb02, Fuxman-vldb06, Gottlob-jacm10, glav-sigmod11, eirene-pvldb11, mweaver-sigmod12}.
\textsc{Clio}~\cite{clio-vldb02, Fuxman-vldb06}  infers schema mappings for relational and XML schemas given a value correspondence between atomic schema elements. Another line of work~\cite{Atzeni-vldbj08,Papotti-jwe05} uses model management operators such as ModelGen~\cite{modelgen-cidr03} to translate schemas from one model to another. In contrast, \toolname takes  examples as input, which are potentially easier to construct for non-experts. There are also several other schema mapping techniques that  use examples.
For instance,  \eirene~\cite{glav-sigmod11, eirene-pvldb11} interactively solicits examples to generate a GLAV specification.  \eirene  is restricted to relational-to-relational mappings and does not support data filtering, but their language can also express mappings that are not expressible in the Datalog fragment used in this work. 
Similarly, Bonifati et al. use example tuples to infer possible schema mappings and interact with the user via binary questions to refine the inferred mappings~\cite{Bonifati-sigmod17}. In contrast to \toolname, \cite{Bonifati-sigmod17} only focuses on relational-to-relational mappings~\cite{glav-sigmod11}.
Finally, \mweaver~\cite{mweaver-sigmod12}  provides a GUI to help  users to interactively generate attribute correspondences  based on  examples.  \mweaver is also restricted  to relational-to-relational mappings and disallows numerical attributes in the source database for performance reasons. Furthermore, it requires the entire source  instance to perform  mapping inference.

\vspace{3pt}
\noindent \emph{\textbf{Program synthesis for data transformation.}}
There has  been significant work on automating data transformations using program synthesis~\cite{mitra-vldb18, morpheus-pldi17, trinity-vldb19, progFromEx-pldi11, hades-pldi16, scythe-pldi17, blinkfill-vldb16, rulesynth-vldb17}. Many techniques focus purely on table or string transformations~\cite{morpheus-pldi17, progFromEx-pldi11, scythe-pldi17, trinity-vldb19, blinkfill-vldb16}, whereas {\sc Hades}~\cite{hades-pldi16} (resp. \mitra~\cite{mitra-vldb18}) focuses on document-to-document (resp. document-to-table) transformations. Our technique generalizes prior work  by automating transformations between many different types of database schemas. Furthermore, as we demonstrate in Section~\ref{sec:mitra-comp}, this generality does not come at the cost of practicality, and, in fact, performs faster synthesis.

\vspace{3pt}
\noindent \emph{\textbf{Inductive logic programming.}}
Our work is related to inductive logic programming (ILP) where the goal is  to synthesize a logic program  consistent with a   set of examples ~\cite{Muggleton-ml14, metagol-ml15, ntp-nips17, ilp-cacm15, Lin-ecai14, neuralLP-nips17, partialILP-ijcai18, difflog-ijcai19}.  Among ILP techniques, our work is most similar to  recent work on Datalog program synthesis~\cite{zaatar-cp17, alps-fse18}. In particular, \zaatar~\cite{zaatar-cp17} encodes an under-approximation of Datalog semantics using the theory of arrays and reduces   synthesis to SMT solving.
However, this technique imposes an upper bound on the number of clauses and atoms  in the Datalog program.
The \alps tool~\cite{alps-fse18} also performs Datalog program synthesis from examples but additionally requires  meta-rule templates. In contrast, our  technique focuses on a recursion-free subset of Datalog, but it does not require additional user input beyond examples and learns from failed synthesis attempts by using the concept of minimal distinguishing projections.

\vspace{3pt}
\noindent \emph{\textbf{Learning from conflicts in synthesis.}}
Our method bears  similarities to recent work on conflict-driven learning in program synthesis~\cite{neo-pldi18, trinity-vldb19, migrator-pldi19} where the goal is to learn useful information from failed synthesis attempts. For example, \neo~\cite{neo-pldi18} and \trinity~\cite{trinity-vldb19} use component specifications to perform root cause analysis and identify other programs that also cannot satisfy the specification. \toolname's sketch completion approach is based on a similar insight, but it uses  Datalog-specific techniques to perform inference. Another related work is \migrator~\cite{migrator-pldi19}, which automatically synthesizes a new version of a SQL program given its old version and a new relational schema. In contrast to \migrator, our method addresses the problem of migrating \emph{data} rather than \emph{code} and is not limited to relational schemas. In addition, while \migrator also aims to learn from failed synthesis attempts, it does so using testing as opposed to MDP analysis for Datalog.

\vspace{3pt}
\noindent \emph{\textbf{Universal and core solutions for data exchange.}}
Our work is related to the data exchange problem~\cite{Fagin-tcs05}, where the goal is to construct a target instance $J$ given a source instance $I$ and a schema mapping $\Sigma$ such that $(I, J) \models \Sigma$. Since such a solution $J$ is not unique, researchers have developed the concept of universal and core solutions to characterize generality and compactness~\cite{Fagin-tcs05,Fagin-tods05}.
In contrast, during its data migration phase, \toolname obtains a unique target instance by executing the synthesized Datalog program on the source instance. The target instance generated by \toolname is the least Herbrand model of the Datalog rules and the source instance~\cite{datalog-tkde89}. While the least Herbrand model also characterizes generality and compactness of the target instance, the relationship between the least Herbrand model and the universal/core solution for data exchange requires further theoretical investigation.


 \vspace{-5pt}
\section{Limitations} \label{sec:limit}

Our approach has three limitations that we plan to address in future work. First, our synthesis technique does not provide any guarantees about the optimality of the synthesized Datalog programs, either in terms of performance or size.  Second, we assume that the  examples provided by the user are always correct; thus, our method does not handle any noise in the specification. 
Third, we assume that we can compare string values for equality when inferring the attribute mapping and obtain the proper matching using set containment. If the values are slightly changed, or if there is a different matching heuristic between attributes, our technique would not be able to synthesize the desired program. However, this shortcoming can be overcome by more sophisticated schema matching techniques~\cite{Doan-sigmod01,cupid-vldb01}.

\section{Conclusion} \label{sec:concl}

We  have proposed a new PBE technique that can synthesize Datalog programs to automate  data migration tasks. 
We evaluated our tool, \toolname, on $28$ benchmarks that involve migrating data between  different types of database schemas and showed that \toolname can successfully automate the desired data migration task from small input-output examples.

\section{Acknowledgments}

We would like to thank the anonymous reviewers and our shepherd for their useful comments and suggestions. This work is supported by NSF Awards \#1712067 and \#1762299.

\bibliographystyle{abbrv}
\bibliography{main}
\balance


\ifextended{
\appendix
\section{Theorems and Proofs}

\begin{definition}[Equivalence of Datalog Programs] \label{def:datalog-equiv}
Given two Datalog programs $\prog_1$ and $\prog_2$, $\prog_1$ and $\prog_2$ are said to be \emph{equivalent}, denoted by $\prog_1 \simeq \prog_2$, if and only if $\denot{\prog_1}_{\ein} = \denot{\prog_2}_{\ein}$ holds for any input instance $\ein$, i.e.
\[
    \prog_1 \simeq \prog_2 \triangleq \forall \ein.~ \denot{\prog_1}_{\ein} = \denot{\prog_2}_{\ein}
\]
\end{definition}

\begin{proof}[of Theorem~\ref{thm:substitution}]
Suppose $\denot{\prog} = \forall X. \Phi(X)$. Since $\hat{\sigma}$ is an injective mapping from variables $X$ to variables $Y$, we have $\denot{\prog\hat{\sigma}} = (\forall X.\Phi(X))[Y/X] = \forall Y. \Phi(Y) = \denot{\prog}$. Based on the semantics of Datalog programs, we have program $\prog$ is equivalent to $\prog\hat{\sigma}$, i.e. $\prog \simeq \prog\hat{\sigma}$.
\end{proof}

\begin{lemma}[Correctness of Generalize] \label{lem:generalize-correctness}
Given two models $\model, \model'$ with the same domain, if $\model' \models Generalize(\model)$, then (1) $\model' = \model \hat{\sigma}$ where $\hat{\sigma}$ is an injective substitution, and (2) for any $x \in \dom(\model)$, if $\sigma(x)$ corresponds to a head variable, then $\model(x) = \model'(x)$.
\end{lemma}
\begin{proof}
Since (2) holds by the definition of $\emph{Generalize}(\model)$, let us focus on property (1). Given that $\model$ and $\model'$ have the same domain and $\model' = \model \hat{\sigma}$, we only need to prove $\hat{\sigma}$ is injective. Let us prove it by contradiction. Suppose $\hat{\sigma}$ is not injective and maps two different values $v_1, v_2 (v_2 \neq v_1)$ in the range of $\model$ to one single value $v_3$, i.e. $\hat{\sigma}(v_1) = \hat{\sigma}(v_2) = v_3$. Since $\model' = \model \hat{\sigma}$, we have $v_3$ in the range of $\model'$. On one hand, suppose $\model(x_1) = v_1, \model(x_2) = v_2$, it follows that $\model'(x_1) = \model(x_1) \hat{\sigma} = v_3$ and $\model'(x_2) = \model(x_2) \hat{\sigma} = v_3$, so we have $\model'(x_2) = \model'(x_1)$. On the other hand, we know $\model'(x_2) \neq \model'(x_1)$ from the definition of $\emph{Generalize}(\model)$ because $v_1 \neq v_2$, which is a contradiction.
\end{proof}

\begin{definition}[Projection] \label{def:projection}
Let $R$ be a non-recursive rule with head $H(X)$ where $X$ is a set of variables that correspond to attributes $A$ of relation $H$. Assume that $L \subseteq A$ is a subset of attributes in $H$ and $X_L$ are their corresponding variables, rule $R(L)$ that substitutes head $H(X)$ with $H(X_L)$ in $R$ is said to be the \emph{projection} of $R$.
\end{definition}

\begin{lemma}[Projection Property] \label{lem:projection}
Assuming $R$ is a non-recursive rule and $R(L)$ is its projection on attributes $L$, it holds that $\denot{R(L)}_{\ein} = \Pi_{L}(\denot{R}_{\ein})$ for any input $\ein$.
\end{lemma}
\begin{proof}
Without loss of generality, suppose rule $R$ is of the form $H(X) \colondash B(X, Y)$. Then $\denot{R} = \forall X, Y.~ B(X, Y) \rightarrow H(X)$. Based on the definition of $R(L)$, rule $R(L)$ is of the form $H(X_L) \colondash B(X, Y)$, so $\denot{R(L)} = \forall X, Y. B(X, Y) \rightarrow H(X_L)$. According to the semantics of Datalog, we have $\denot{R(L)}_{\ein} = \Pi_L(\denot{R}_{\ein})$ for any input $\ein$.
\end{proof}

\begin{lemma}[Correctness of Generalize with MDP] \label{lem:generalize-mdp-correctness}
Consider a program sketch $\sketch$, two models $\model, \model'$, and an MDP $\mdp$. Let $\prog = \textsf{Instantiate}(\sketch, \model)$ and $\prog' = \textsf{Instantiate}(\sketch, \model')$. If $\model' \models Generalize(\model, \mdp)$, then $\Pi_{\mdp}(\denot{\prog}_{\ein}) = \Pi_{\mdp}(\denot{\prog'}_{\ein})$ holds for any input $\ein$.
\end{lemma}
\begin{proof}
Let $\eout = \denot{\prog}_{\ein}$, $\eout' = \denot{\prog'}_{\ein}$ and consider the projection $\prog(\mdp)$ of $\prog$ and projection $\prog'(\mdp)$ of $\prog'$. Observe the difference between $\emph{Generalize}(\model, \mdp)$ and $\emph{Generalize}(\model)$, we know that projection $\prog(\mdp) = \textsf{Instantiate}(\sketch, \model)$ and projection $\prog'(\mdp) = \textsf{Instantiate}(\sketch, \model')$. Furthermore, applying $\emph{Generalize}(\model, \mdp)$ on program $\prog$ is equivalent to applying $\emph{Generalize}(\model)$ on program $\prog(\mdp)$, because the head variables of $\prog(\mdp)$ is the same as variables corresponding to attributes in $\mdp$ for $\prog$. Then by Lemma~\ref{lem:generalize-correctness}, it holds that (1) $\model' = \model \hat{\sigma}$ where $\hat{\sigma}$ is an injective substitution, and (2) for any $x \in \dom(\model)$, if $\model(x)$ corresponds to a head variable of $\prog(\mdp)$, then $\model(x) = \model'(x)$. Thus, given that $\prog'(\mdp) = \prog(\mdp) \hat{\sigma}$, we know $\prog(\mdp) \simeq \prog'(\mdp)$ by Theorem~\ref{thm:substitution}. In other words, $\denot{\prog(\mdp)}_{\ein} = \denot{\prog'(\mdp)}_{\ein}$ holds for any input $\ein$. Therefore, by Lemma~\ref{lem:projection}, we have $\Pi_{\mdp}(\denot{\prog}_{\ein}) = \denot{\prog(\mdp)}_{\ein} = \denot{\prog'(\mdp)}_{\ein} = \Pi_{\mdp}(\denot{\prog'}_{\ein})$ for any input $\ein$.
\end{proof}

\begin{lemma}[Correctness of \textsc{MDPSet}] \label{lem:mdpset-correctness}
Given a Datalog program $\prog$ and an input-output example $(\ein, \eout)$, suppose $\denot{\prog}_{\ein} = \eout'$ and $\eout' \neq \eout$. If $\mdpSet = \textsc{MDPSet}(\eout', \eout)$, then for every element $\mdp \in \mdpSet$, it holds that $\mdp$ is an MDP for program $\prog$ and example $(\ein, \eout)$.
\end{lemma}
\begin{proof}
Observe that Lines 2 -- 5 of the \textsc{MDPSet} procedure in Algorithm~\ref{algo:mdpset} establish the facts: (1) $\mdpSet = \emptyset$, (2) $\visited = \set{\set{a} ~|~ a \in \emph{Attrs}(\eout')}$, and (3) $\worklist$ is a queue initialized to contain all the elements in $\visited$. Now consider the following invariants of the while loop at Lines 6 -- 14:
\begin{itemize}
    \item[$I_1$.] $\forall \mdp \in \mdpSet.~\Pi_{\mdp}(\eout) \neq \Pi_{\mdp}(\denot{\prog}_{\ein})$
    \item[$I_2$.] $\forall \mdp \in \mdpSet.~\forall A \subset \mdp.~\Pi_A(\eout) = \Pi_A(\denot{\prog}_{\ein})$
\end{itemize}
Then let us prove they are loop invariants in more detail:
\begin{itemize}
\item
$I_1$ trivially holds at the beginning of the loop, because $\mdpSet$ is empty. To see why $I_1$ holds after every iteration of the loop, observe Line 8 and Line 14 of the algorithm. A new element $L$ can only be added to $\mdpSet$ if $\Pi_L(\eout') \neq \Pi_L(\eout)$. Since $\denot{\prog}_{\ein} = \eout'$, we have $\Pi_L(\eout) \neq \Pi_L(\denot{\prog}_{\ein})$. Thus, $I_1$ holds after every iteration of the loop.
\item
$I_2$ trivially holds at the beginning of the loop, because $\mdpSet$ is empty. Let us prove $I_2$ holds after every iteration of the loop. For an element $L$ to be added to $\mdpSet$ at Line 14, we only need to prove $\forall L' \subset L.~\Pi_{L'}(\eout) = \Pi_{L'}(\denot{\prog}_{\ein})$. We now prove it by contradiction. Suppose there exists an attribute set $A \subset L$ such that $\Pi_A(\eout) \neq \Pi_A(\denot{\prog}_{\ein})$. (i) If $A \in \mdpSet$, then $L$ should not be added to $\mdpSet$ by the condition at Line 14. Contradiction. (ii) If $A \notin \mdpSet$, it must hold that $\exists A' \subset A.~A' \in \mdpSet$ by the condition at Line 14. Since $A' \subset A \subset L$ and $A' \in \mdpSet$, we know $L$ should not be added to $\mdpSet$. Contradiction.
\end{itemize}
Based on the definition of MDP and loop invariants $I_1$ and $I_2$, we have proved that every element $\mdp \in \mdpSet$ is an MDP for program $\prog$ and example $(\ein, \eout)$.
\end{proof}

\begin{proof}[of Theorem~\ref{thm:prune-safety}]
Given an input-output example $(\ein, \eout)$ and a candidate program $\prog$ with $\denot{\prog}_{\ein} = \eout'$, observe that the \textsc{Synthesize} procedure in Algorithm~\ref{algo:synthesis} can only reach \textsc{Analyze} when  $\eout' \neq \eout$. Thus by Lemma~\ref{lem:mdpset-correctness}, the return value $\mdpSet$ of \textsc{MDPSet} in Algorithm~\ref{algo:analyze} contains a set of MDPs. For each MDP $\mdp \in \mdpSet$, observe that $\psi = \emph{Generalize}(\model, \mdp)$ according to Lines 5 -- 10 of the \textsc{Analyze} procedure in Algorithm~\ref{algo:analyze}, so it holds that $\model' \models \emph{Generalize}(\model, \mdp)$ for a model $\model'$ corresponding to formula $\psi$. Let us denote the program corresponding to $\model'$ by $\prog'$. Then by Lemma~\ref{lem:generalize-mdp-correctness}, we have $\Pi_{\mdp}(\denot{\prog}_{\ein}) = \Pi_{\mdp}(\denot{\prog'}_{\ein})$. Since $\denot{\prog}_{\ein} = \eout'$ and $\eout' \neq \eout$, we know $\denot{\prog'}_{\ein} \neq \eout$, i.e. $\prog'$ is an incorrect program. Given that the return value of \textsc{Analyze} is the conjunction of all $\neg \psi$, we have proved the theorem.
\end{proof}

\begin{definition}[Sketch Completion] \label{def:sketch-completion}
Given a sketch $\sketch$ with holes $\vec{\emph{\hole}}$, a program $\prog$ is said to be the \emph{completion} of $\sketch$, denoted by $\prog \in \gamma(\sketch)$, if $\prog = \sketch[\vec{c}/\vec{\emph{\hole}}]$ where $c_i$ is a constant in the domain of $\emph{\hole}_i$.
\end{definition}

\begin{lemma}[Correctness of GenIntensionalPreds] \label{lem:intensional-correctness}
Given a top-level record type $\name$ in target schema $\schema'$, invoking $\textsc{GenIntensionalPreds}(\schema', \name)$ generates the following relations:
\begin{enumerate}
    \item relation $R_{\name}(v_{a_1}, \ldots, v_{a_n})$ for record type $\name$, where $v_{a_i} $ is the variable corresponding to attribute $a_i$ of $\name$
    \item relation $R_{\name'}(v_{\name'}, v_{a_1}, \ldots, v_{a_n})$ for each nested record type $\name'$ in $\name$, where $v_{\name'}$ is the variable corresponding to $\name'$ and $v_{a_i} $ corresponds to attribute $a_i$ of $\name'$
\end{enumerate}
\end{lemma}
\begin{proof}
Prove by structural induction on name $\name$.
\begin{itemize}
\item Base case.
$\name$ is a primitive attribute, i.e. $\schema'(\name) \in \emph{PrimType}$. According to the \textrm{InPrim} rule in Figure~\ref{fig:rules-intensional}, \textsc{GenIntensionalPreds} generates a variable $v_{\name}$ for $\name$ and no head relations.
\item Inductive case.
$\name$ is a record type with attributes $\name_1, \ldots, \name_m$, i.e. $\schema'(\name) = \set{\name_1, \ldots, \name_m}$. Suppose for every name $\name_j$ with $j \in [1, m]$, \textsc{GenIntensionalPreds} generates $R_{\name_j}(v_{a_1}, \ldots, v_{a_n})$ if $\name_j$ is a top-level record type and otherwise generates relation $R_{\name'_j}(v_{\name'_j}$, $v_{a_1}, \ldots, v_{a_n})$ for each nested record type $\name'_j$ in $\name_j$, we need to prove the property holds for $\name$. We perform a case split on name $\name$.
\begin{enumerate}
    \item[(1)] If $\name$ is a top-level record type. The property simply holds by the \textrm{InRec} rule in Figure~\ref{fig:rules-intensional}.
    \item[(2)] If $\name$ is not a top-level record type. The property holds by the \textrm{InRecNested} rule in Figure~\ref{fig:rules-intensional}.
\end{enumerate}
\end{itemize}
By the principle of structural induction, we have proved this lemma.
\end{proof}

\begin{lemma}[Correctness of GenExtensionalPreds] \label{lem:extensional-correctness}
Given a sequence of record types $\name_1, \ldots, \name_m$ in source schema $\schema$, where $\name_{i+1}$ is nested in $\name_i$ for $i \in [1, m-1]$ and $\name_1$ is at the top level, invoking \textsc{GenExtensionalPreds} $\!\!(\schema, \name_m)$ generates the following relations:
\begin{enumerate}
    \item relation $R_{\name_1}(h_{a_1}, \ldots, h_{a_n})$, where $h_{a_i}$ is the hole corresponding to attribute $a_i$ of $N_1$
    \item relation $R_{\name_j}(v_{\name_j}, h_{a_1}, \ldots, h_{a_n})$ for each $\name_j$ with $j \in [2, m]$, where $v_{\name_j}$ is the variable corresponding to record type $\name_j$, and $h_{a_i}$ is the hole corresponding to attribute $a_i$ of $N_j$
\end{enumerate}
\end{lemma}
\begin{proof}
Prove by induction on $m$.
\begin{itemize}
\item Base case.
$m = 1$. In this case, we only have one top-level record $\name_1$. According to the rules \textrm{ExRec} and \textrm{ExPrim} in Figure~\ref{fig:rules-extensional}, \textsc{GenExtensionalPreds} generates a relation $R_{\name_1}(h_{a_1}, \ldots, h_{a_n})$, where $h_{a_i}$ is the hole corresponding to attribute $a_i$ of $N_1$.
\item Inductive case.
Suppose the property holds for $m = k-1$, let us prove it holds for $m = k$. By inductive hypothesis, we know that \textsc{GenExtensionalPreds} generates a relation $R_{\name_1}(h_{a_1}, \ldots, h_{a_n})$ for $\name_1$ and generates $R_{\name_j}(v_{\name_j}, h_{a_{j_1}}, \ldots, h_{a_{j_n}})$ for each $j \in [2, k-1]$. According to rules \textrm{ExRecNested} and \textrm{ExPrim} in Figure~\ref{fig:rules-extensional}, it also generates a relation $R_{\name_k}(v_{\name_k}, h_{a_{k_1}}, \ldots, h_{a_{k_n}})$. Thus the property holds. 
\end{itemize}
By the principle of induction, we have proved this lemma.
\end{proof}

\begin{lemma}[Sketch of Rules] \label{lem:sketch-rule}
Given a source and target schema $\schema, \schema'$, an example $\ex = (\ein, \eout)$, and a top level record type $\name$ in $\schema'$, let $\attrMap = \textsc{InferAttrMapping}(\schema, \schema', \ex)$ and sketch $\sketch = \textsc{GenRuleSketch}(\attrMap, \schema, \schema', \name)$. If a Datalog rule $R$ satisfies the following conditions:
\begin{enumerate}
    \item every relation corresponding to a record type nested in $\name$ occurs exactly once in the head of $R$
    \item all body relations correspond to record types in $\schema$
    \item $\denot{R}_{\ein} = \eout_{\name}$ where $\eout_{\name}$ only contains the output for record type $N$ in $\eout$
\end{enumerate}
then there exists a program $R' \in \gamma(\sketch)$ such that $R' \simeq R$.
\end{lemma}
\begin{proof}
First of all, let us consider the head of rule $R$. Since the head of $R$ contains exactly one relation for each nested record in $\name$ and $\denot{R}_{\ein} = \eout_{\name}$, so the arguments of these relations should cover all the attributes in $\name$. In addition, considering the migration procedure in Section~\ref{sec:prelim-migration}, each nested record type should have an extra attribute to track its parent. Thus, the head of $R$ contains (1) a relation $R_{\name}$ for $\name$ and its arity is the number of attributes in $\name$, and (2) relation $R_{\name'}$ for each nested record type in $\name$ and its arity is one more than the number of attributes in $\name'$. By Lemma~\ref{lem:intensional-correctness}, we know the head of $\sketch$ is identical to the head of $R$ up to variable renaming, and let us denote the variable renaming by $\hat{\sigma}$.
Without loss of generality, suppose $\hat{\sigma}$ is enforced to be applied on $R$, i.e. the head variable for target attribute $a$ is renamed to $v_a$ and so does the variables in the body. Then the head of $R$ and $\sketch$ are syntactically identical, say $H$, and we will use $H$ as the head of $R'$.

Next let us focus on the body of $R$. Observe that the natural bound on the number of copies for one source attribute is the number of aliasing attributes in the source and target schema, i.e. the maximum number of copies for attribute $a$ in record type $\name'$ is bounded by $|\set{a' ~|~ a' \in \emph{PrimAttrbs}(\name') \land a' \in \attrMap(a)}|$. Since adding more copies beyond it would eventually boil down to renaming variables of existing predicates, we will prove there exists $R' \in \gamma(\sketch)$ such that $R' \simeq R$ under this natural bound by induction.
\begin{itemize}
\item Base case.
$R$ has only one relation in the body. Suppose $R$ is of the form $H(v_{a'_1}, \ldots, v_{a'_m}) \colondash B(v_1, \ldots, v_n)$ where $B$ corresponds to record type $\name_B$ in $\schema$ and $v_i$ is the variable for attribute $a_i$. Since $\denot{R}_{\ein} = \eout_{\name}$, for every attribute $a'_j$ in $\schema'$, there exists an attribute $a_k$ in $\schema$ such that $\Pi_{a_k}(\ein) = \Pi_{a'_j}(\eout_{\name})$. Thus, we have $a'_j \in \attrMap(a_k)$ based on the definition of \textsc{InferAttrMapping}.
Observing Lines 9 -- 12 of Algorithm~\ref{algo:sketch-gen} and combining with Lemma~\ref{lem:extensional-correctness}, we know that \textsc{GenRuleSketch} generates at least $m$ copies of predicate $B(h_1, \ldots, h_n)$. Then according to Lines 13 -- 18 of Algorithm~\ref{algo:sketch-gen}, for any attribute $a_i \in \dom(\attrMap)$, $a'_j$ is included in the domain of $h_i$. Furthermore, for any attribute $a_i \not\in \dom(\attrMap)$, a fresh variable $v_{a_i}^1$ is in the domain of $h_i$. Thus we can construct the body of $R'$ such that $R' \in \gamma(\sketch)$, by instantiating $h_i$ with $v_{a'_j}$ if $a_i \in \dom(\attrMap)$ and instantiating $h_i$ with $v_{a_i}^1$ if $a_i \in \dom(\attrMap)$. In this way, $R'$ is equivalent to $R$ up to variable renaming.
\item Inductive case.
Assuming there exists $R'_k \in \gamma(\sketch)$ with at most $k$ body predicates such that $R'_k \simeq R_k$ if $R_k$ is of the form $H \colondash B_1, \ldots, B_k$, let us prove there exists $R'_{k+1} \in \gamma(\sketch)$ with at most $k+1$ body predicates such that $R'_{k+1} \simeq R_{k+1}$ if $R_{k+1}$ is of the form $H \colondash B_1, \ldots, B_k, B_{k+1}$.
For the purpose of illustration, suppose $R_{k+1}$ is of the form $H(v_{a'_1}, \ldots, v_{a'_m}) \colondash B_1, \ldots$, $B_k, B_{k+1}(v_1, \ldots, v_n)$, where $B_i$ corresponds to record type $\name_i$ in $\schema$ for $i \in [1, k+1]$, and $v_j$ is the variable for attribute $a_j$ for $j \in [1, n]$. Note that $\name_{i_1}$ and $\name_{i_2}$ may refer to the same record type in general. Without loss of generality, we assume $N_{k+1}$ is record type $N_B$.
By inductive hypothesis, we know that $R'_k \simeq R_k$ with at most $k$ predicates in the body. To construct the new rule $R'_{k+1}$, we only need to instantiate one more predicate $B(h_1, \ldots, h_n)$ and add it to the body of $R'_k$, where $B$ is the relation corresponding to $N_B$. And such predicate sketch $B(h_1, \ldots, h_n)$ exists in the body of $\sketch$ by Lemma~\ref{lem:extensional-correctness}. To see how to construct the variable $v'_i$, let us consider the variable $v_i$ in $B_{k+1}(v_1, \ldots, v_n)$ and its corresponding attribute $a_i$ in schema $\schema$.
\begin{enumerate}
    \item[(1)] $v_i$ is a fresh variable. We can simply pick a variable $v_{a_i}^j$ that does not occur in $R'_k$.
    \item[(2)] $v_i$ occurs in the head and body predicates $B_1, \ldots, B_k$. Suppose the corresponding head attribute is $a'_r$, we know $\Pi_{a_i}(\ein) = \Pi_{a'_r}(\eout_N)$ from $\denot{R_{k+1}}_{\ein} = \eout_{\name}$. Based on the definition of \textsc{InferAttrMapping}, we have $a'_r \in \attrMap(a_i)$, so $v_{a'_r}$ is in the domain of $h_i$ according to Lines 9 -- 18 of Algorithm~\ref{algo:sketch-gen} and Lemma~\ref{lem:extensional-correctness}. Thus, we can use $v_{a'_r}$ to instantiate $h_i$.
    \item[(3)] $v_i$ does not occur in the head but occurs in body predicates $B_1, \ldots, B_k$. Suppose the corresponding source attributes are $a_{j_1}, \ldots, a_{j_r}$. Since $\denot{R_{k+1}}_{\ein} = \eout_{\name} = \denot{R_k}_{\ein}$, we have $\Pi_{a_{j_s}}(\ein) \subseteq \Pi_{a_i}(\ein)$ or $\Pi_{a_{j_s}}(\ein)$ $\supseteq \Pi_{a_i}(\ein)$ for all $s \in [1, r]$. So by the definition of \textsc{InferAttrMapping}, we have $a_i \in \attrMap(a_{j_s})$ or $a_{j_s} \in \attrMap(a_i)$. Thus, $v_{a_{j_s}}$ is in the domain of $h_i$ or $v_i$ is in the domain of the hole corresponding to attribute $a_{j_s}$. Therefore, we can instantiate these holes accordingly.
\end{enumerate}
So far, we have constructed the new rule $R'_{k+1}$ such that $R'_{k+1} \in \gamma(\sketch)$ and $R'_{k+1}$ is equivalent to $R_{k+1}$ up to variable renaming.
\end{itemize}
By the principle of induction, we have proved this lemma.
\end{proof}

\input{fig-sen-detail}

\begin{figure*}[!htb]
\centering
    \begin{subfigure}[b]{0.23\textwidth}
        \centering
        \begin{tikzpicture}[scale=0.7]
        \begin{axis}[
            width=2.2in,
            x label style={below=-5pt, font=\small},
            y label style={below=8pt, font=\small},
            axis y line*=left,
            xmin=0,
            xtick={0,2,...,11},
            ymin=0,
            ymax=10,
        	xlabel={\# Examples},
        	ylabel={Synthesis Time (s)},
        ]
        \addplot[color=blue,mark=*,width=2pt]
            table[x=a, y=b] {bench25-time.data};
        \label{legend:bench25-time}
        \end{axis}

        \begin{axis}[
            width=2.2in,
            legend cell align=left,
            legend style={at={(0.5,1.17)}, font=\small, legend columns=-1, anchor=north},
            y label style={below=160pt, font=\small},
            legend image post style={scale=0.5},
            hide x axis,
            axis y line*=right,
            xmin=0,
            xtick={0,2,...,11},
            ymin=0,
            ymax=110,
        	ylabel={Success Rate (\%)},
        ]
        \addlegendimage{/pgfplots/refstyle=legend:bench25-time}\addlegendentry{Time(s)}
        \addplot[color=red,mark=square*,width=2pt]
            table[x=a, y=b] {bench25-rate.data};
        \addlegendentry{Succ Rate(\%)}
        \end{axis}
    \end{tikzpicture}
    \vspace{-15pt}
    \caption{MLB-3}
    \vspace{-3pt}
    \end{subfigure}
    \begin{subfigure}[b]{0.23\textwidth}
        \centering
        \begin{tikzpicture}[scale=0.7]
        \begin{axis}[
            width=2.2in,
            x label style={below=-5pt, font=\small},
            y label style={below=8pt, font=\small},
            axis y line*=left,
            xmin=0,
            xtick={0,2,...,11},
            ymin=0,
            ymax=0.6,
        	xlabel={\# Examples},
        	ylabel={Synthesis Time (s)},
        ]
        \addplot[color=blue,mark=*,width=2pt]
            table[x=a, y=b] {bench26-time.data};
        \label{legend:bench26-time}
        \end{axis}

        \begin{axis}[
            width=2.2in,
            legend cell align=left,
            legend style={at={(0.5,1.17)}, font=\small, legend columns=-1, anchor=north},
            y label style={below=160pt, font=\small},
            legend image post style={scale=0.5},
            hide x axis,
            axis y line*=right,
            xmin=0,
            xtick={0,2,...,11},
            ymin=0,
            ymax=110,
        	ylabel={Success Rate (\%)},
        ]
        \addlegendimage{/pgfplots/refstyle=legend:bench26-time}\addlegendentry{Time(s)}
        \addplot[color=red,mark=square*,width=2pt]
            table[x=a, y=b] {bench26-rate.data};
        \addlegendentry{Succ Rate(\%)}
        \end{axis}
    \end{tikzpicture}
    \vspace{-15pt}
    \caption{Airbnb-3}
    \vspace{-3pt}
    \end{subfigure}
    \begin{subfigure}[b]{0.23\textwidth}
        \centering
        \begin{tikzpicture}[scale=0.7]
        \begin{axis}[
            width=2.2in,
            x label style={below=-5pt, font=\small},
            y label style={below=8pt, font=\small},
            axis y line*=left,
            xmin=0,
            xtick={0,2,...,11},
            ymin=0,
            ymax=6,
        	xlabel={\# Examples},
        	ylabel={Synthesis Time (s)},
        ]
        \addplot[color=blue,mark=*,width=2pt]
            table[x=a, y=b] {bench27-time.data};
        \label{legend:bench27-time}
        \end{axis}

        \begin{axis}[
            width=2.2in,
            legend cell align=left,
            legend style={at={(0.5,1.17)}, font=\small, legend columns=-1, anchor=north},
            y label style={below=160pt, font=\small},
            legend image post style={scale=0.5},
            hide x axis,
            axis y line*=right,
            xmin=0,
            xtick={0,2,...,11},
            ymin=0,
            ymax=110,
        	ylabel={Success Rate (\%)},
        ]
        \addlegendimage{/pgfplots/refstyle=legend:bench27-time}\addlegendentry{Time(s)}
        \addplot[color=red,mark=square*,width=2pt]
            table[x=a, y=b] {bench27-rate.data};
        \addlegendentry{Succ Rate(\%)}
        \end{axis}
    \end{tikzpicture}
    \vspace{-15pt}
    \caption{Patent-3}
    \vspace{-3pt}
    \end{subfigure}
    \begin{subfigure}[b]{0.23\textwidth}
        \centering
        \begin{tikzpicture}[scale=0.7]
        \begin{axis}[
            width=2.2in,
            x label style={below=-5pt, font=\small},
            y label style={below=8pt, font=\small},
            axis y line*=left,
            xmin=0,
            xtick={0,2,...,11},
            ymin=0,
            ymax=7,
        	xlabel={\# Examples},
        	ylabel={Synthesis Time (s)},
        ]
        \addplot[color=blue,mark=*,width=2pt]
            table[x=a, y=b] {bench28-time.data};
        \label{legend:bench28-time}
        \end{axis}

        \begin{axis}[
            width=2.2in,
            legend cell align=left,
            legend style={at={(0.5,1.17)}, font=\small, legend columns=-1, anchor=north},
            y label style={below=160pt, font=\small},
            legend image post style={scale=0.5},
            hide x axis,
            axis y line*=right,
            xmin=0,
            xtick={0,2,...,11},
            ymin=0,
            ymax=110,
        	ylabel={Success Rate (\%)},
        ]
        \addlegendimage{/pgfplots/refstyle=legend:bench28-time}\addlegendentry{Time(s)}
        \addplot[color=red,mark=square*,width=2pt]
            table[x=a, y=b] {bench28-rate.data};
        \addlegendentry{Succ Rate(\%)}
        \end{axis}
    \end{tikzpicture}
    \vspace{-15pt}
    \caption{Bike-3}
    \vspace{-3pt}
    \end{subfigure}

\vspace{10pt}
\caption{Detailed sensitivity analysis (continued)}
\label{fig:sen-detail-2}
\end{figure*}

\begin{theorem}[Sketch Property] \label{thm:sketch-prop}
Given a source and target schema $\schema, \schema'$, and an example $\ex = (\ein, \eout)$, let $\attrMap = \textsc{InferAttrMapping}(\schema, \schema', \ex)$ and sketch $\sketch = \textsc{SketchGen}$ $(\attrMap, \schema, \schema')$. If a Datalog program $\prog$ conforms to the syntax shown in Figure~\ref{fig:syntax-datalog} and $\prog$ satisfies the following conditions:
\begin{enumerate}
    \item for each record type $\name$ in $\schema'$, there is exactly one rule in $\prog$ such that its head corresponds to $\name$
    \item for each rule $R$ in $\prog$, the body relations of $R$ correspond to record types in $\schema$
    \item $\denot{\prog}_{\ein} = \eout$
\end{enumerate}
then there exists a program $\prog' \in \gamma(\sketch)$ such that $\denot{\prog'}_{\ein} = \eout$.
\end{theorem}
\begin{proof}
Observe that the \textsc{SketchGen} procedure simply invokes the \textsc{GenRuleSketch} for each top-level record type $N$ in target schema $\schema'$. This theorem directly follows from Lemma~\ref{lem:sketch-rule}.
\end{proof}

\vspace{5pt}
\section{Interactive Mode}
As mentioned in Section~\ref{sec:impl}, \toolname can be used in an interactive mode, where the tool queries the user for more input-output examples if there are multiple Datalog programs that are consistent with the given example. When used in this mode, \toolname does not stop as soon as it finds a single program that satisfies the examples; rather, it keeps searching for additional programs (i.e., sketch completions) until either (a) it  finds another program that is also consistent with the given example, or (b) it proves that there is no other program consistent with the given example (i.e., the SAT encoding from Algorithm 1 becomes unsatisfiable). If there is another program consistent with the example, \toolname uses testing to find a smallest set of input tuples that distinguish the two programs. Specifically, \toolname  enumerates test inputs in increasing order of size where each tuple in the set comes from a validation set sampled from the source database instance. Then, given an input $\ein$ on which the outputs of the two programs differ, \toolname asks the user to provide the corresponding output.


\vspace{5pt}
\section{Sensitivity Analysis}

Our detailed sensitivity analysis for all benchmarks in Section~\ref{sec:eval} is summarized in Figures~\ref{fig:sen-detail} and~\ref{fig:sen-detail-2}. Here, the $x$-axis shows the input size $r$, and the $y$-axis shows both (a) the running time in seconds for each $r$ (the blue line with circles) and (b) the $\%$ of correctly synthesized programs given $r$ randomly-generated records (the red line with squares). To reduce the impact of random errors, we run \toolname $100$ times with different random examples for each $r$. Then, we drop outliers (i.e., fastest and slowest $10\%$) and compute the average  time for the remaining runs.

}\fi

\end{document}